\documentclass[journal,draftcls,onecolumn,12pt,twoside]{IEEEtran}
\usepackage[T1]{fontenc}
\usepackage[latin9]{inputenc}
\usepackage{amsthm}
\usepackage{amsmath}
\usepackage{bbm} 
\usepackage{graphicx} 
\usepackage{color} 
\usepackage{cite}
\usepackage{algpseudocode}
\usepackage{algorithm}
\usepackage{amssymb}
\usepackage{subfigure}
\usepackage{esint}
\usepackage{algpseudocode}
\usepackage{epstopdf}
\usepackage{multirow}
\usepackage{makecell}
\usepackage{float}
\usepackage{lipsum}
\usepackage{balance}
\usepackage{ifpdf}
\usepackage{tikz}
\usepackage{lipsum}
\usepackage{mathtools}
\usepackage{cuted}
\usetikzlibrary{decorations, decorations.text,}
\usepackage{float}
\usepackage{array}
\usepackage[caption=false,font=footnotesize]{subfig}
\usepackage{textcomp}
\usepackage{url}
\usepackage{paralist,esint,amssymb,booktabs,cases,bm,galois}

\newcommand\Tx[1]{\mathrm{#1}}
\newcommand\Se[1]{\mathcal{#1}}
\newcommand\Db[1]{\mathbb{#1}}

\newcommand\Sma[2]{\sum\limits_{#1}^{#2}}
\newcommand\Prd[2]{\prod\limits_{#1}^{#2}}

\newcommand\Nm[1]{\lvert #1\rvert}
\newcommand\Floor[1]{\lfloor #1\rfloor}
\newcommand\Ceil[1]{\lceil #1\rceil}
\newcommand\MB[1]{\left[#1\right]}
\newcommand\LB[1]{\{#1\}}
\newcommand\SB[1]{\left(#1\right)}

\newcommand{\RN}[1]{\textup{\uppercase\expandafter{\romannumeral#1}}}
\usepackage{cleveref}
\usetikzlibrary{calc}
\usetikzlibrary{shadings}
\newtheorem{theo}{Theorem}

\newtheorem{lemma}{Lemma}
\newtheorem{exam}{Example}
\newtheorem{defi}{Definition}
\newtheorem{rem}{Remark}
\newtheorem{cons}{Construction}

\IEEEoverridecommandlockouts

\usepackage{algorithm}
\usepackage{algpseudocode}
\makeatletter
\def\BState{\State\hskip-\ALG@thistlm}
\makeatother

\errorcontextlines\maxdimen
\makeatletter
\newcommand*{\algrule}[1][\algorithmicindent]{\makebox[#1][l]{\hspace*{.5em}\vrule height 0.9 \baselineskip depth 0.3\baselineskip}}%

\newcount\ALG@printindent@tempcnta
\def\ALG@printindent{%
    \ifnum \theALG@nested>0
        \ifx\ALG@text\ALG@x@notext
            \addvspace{-3pt}
        \else
            \unskip
            \ALG@printindent@tempcnta=1
            \loop
                \algrule[\csname ALG@ind@\the\ALG@printindent@tempcnta\endcsname]%
                \advance \ALG@printindent@tempcnta 1
            \ifnum \ALG@printindent@tempcnta<\numexpr\theALG@nested+1\relax
            \repeat
        \fi
    \fi
    }%
\usepackage{etoolbox}
\patchcmd{\ALG@doentity}{\noindent\hskip\ALG@tlm}{\ALG@printindent}{}{\errmessage{failed to patch}}
\makeatother



\begin{document}
\title{Hierarchical Hybrid Error Correction for Time-Sensitive Devices at the Edge}
\author{\IEEEauthorblockN{Siyi Yang, \IEEEmembership{Student Member, IEEE}, Ahmed Hareedy, \IEEEmembership{Member, IEEE}, Robert Calderbank, \IEEEmembership{Fellow, IEEE}, and Lara Dolecek, \IEEEmembership{Senior Member, IEEE}}

\thanks{S. Yang and L. Dolecek are with the Electrical and Computer Engineering Department, University of California, Los Angeles, Los Angeles, CA 90095 USA (e-mail: siyiyang@ucla.edu and dolecek@ee.ucla.edu).}
\thanks{A. Hareedy and R. Calderbank are with the Electrical and Computer Engineering Department, Duke University, Durham, NC 27705 USA (e-mail: ahmed.hareedy@duke.edu and robert.calderbank@duke.edu).}
\thanks{The research was supported in part by the NSF under the grants CCF-BSF 1718389 and CCF 1717602. Part of the paper was presented at the 2019 IEEE Global Conference on Communications (GLOBECOM) \cite{Yang2019HC}.}
}
\maketitle

\begin{abstract} 

Computational storage, known as a solution to significantly reduce the latency by moving data-processing down to the data storage, has received wide attention because of its potential to accelerate data-driven devices at the edge. To meet the insatiable appetite for complicated functionalities tailored for intelligent devices such as autonomous vehicles, properties including heterogeneity, scalability, and flexibility are becoming increasingly important. Based on our prior work on hierarchical erasure coding that enables scalability and flexibility in cloud storage, we develop an efficient decoding algorithm that corrects a mixture of errors and erasures simultaneously. We first extract the basic component code, the so-called extended Cauchy (EC) codes, of the proposed coding solution. We prove that the class of EC codes is strictly larger than that of relevant codes with known explicit decoding algorithms. Motivated by this finding, we then develop an efficient decoding method for the general class of EC codes, based on which we propose the local and global decoding algorithms for the hierarchical codes. Our proposed hybrid error correction not only enables the usage of hierarchical codes in computational storage at the edge, but also applies to any Cauchy-like codes and allows potentially wider applications of the EC codes.
\end{abstract}


\begin{IEEEkeywords}
Distributed Storage, computational storage, hierarchical coding, extended Cauchy codes.
\end{IEEEkeywords}

\IEEEpeerreviewmaketitle

\section{Introduction}
\label{section: introduction}


The burgeoning industry of Internet-of-Things (IoT) is penetrating various companies and solutions, which brings forth an insatiable appetite for intelligent devices at the edge of the Internet. In time-sensitive applications such as autonomous driving, flight control, finance services, etc., milliseconds in latency can make a significant difference in reliability and safety. While storing and processing raw data in the cloud is an immediate solution, moving data over such a distance would be far too slow, not to mention that the disruptive volume of newly generated data adds such a heavy load to the already crowded cloud infrastructure. Computational storage is known for providing in-situ processing, which significantly reduces the latency and saves storage resources by moving data-processing down to the data storage, and is indispensable in building efficient IoT ecosystems \cite{armwebsite,ngdwebsite,sniawebsite}. Various semiconductor companies, including ARM, NIVIDA, NGD Systems, etc., have been exploring the architectures of computational storage. To protect the data against hardware errors, error correction codes (ECCs) are implemented in the solid-state drive (SSD) to produce robust storage units. 


To meet the aggressive latency requirements, techniques performing massive parallel computation such as distributed computing and in-memory computing \cite{sebastian2019computational,yao2020fully,xiao2020analog,kendall2020building,kumar2020third} can be integrated to further reduce the latency. Under this framework, ECCs with hierarchical locality are desired to seamlessly bridge the computing and storage modules. By adopting distributed decoding on the coded data before transmitting them to the computing units, the data are recovered and sent to the processors by small chunks directly in a parallel way to significantly reduce the waiting time cost on decoding the whole block. Codes with hierarchical locality enable the data to be read through a chain of nested sub-blocks with increasing data lengths from top to bottom; this architecture is exploited to increase the overall erasure-correction capability and to reduce the average reading time \cite{huang2017multi,ballentine2018codes,cassuto2017multi,hassner2001integrated}. 


Along with \emph{hierarchical locality} discussed previously, computational storage tailored for intelligent devices is also desired to support heterogeneous, scalable, and flexible resource scheduling of computing cores and storage units under dynamic environments. Take autonomous driving as an example, where an autonomous vehicle is loaded with various tasks of multi-purposes, including object detection, navigation, and path management, etc., which need an SSD of large capacity to efficiently store high definition maps, historical footprints, navigation information, etc., and send this information to computing cores. These tasks typically accommodate nonidentical storage spaces, computing loads, and usage rates, thus naturally supporting \emph{heterogeneity}, allowing nonidentical local data lengths, and providing unequal local protection. \emph{Scalability} enables dynamic allocation of storage resources for individual tasks to accommodate additional workload, i.e., additional cores, without rebuilding the remaining schedule. \emph{Flexibility} has been firstly investigated for dynamic data storage systems in \cite{martnez2018universal}, and it refers to the property that a sector can be split into two smaller blocks without worsening the global error-correction capability nor changing the remaining components. This splitting, for example, is applied when a less used task becomes frequently called.

Various codes offering hierarchical locality have been studied. Cassuto \emph{et al.}\cite{cassuto2017multi} presented the so-called multi-block interleaved codes that provide double-level access; this work introduced the concept of multi-level access. The family of integrated-interleaved (I-I) codes \cite{hassner2001integrated}, including generalized integrated interleaved (GII) codes and extended integrated interleaved (EII) codes, has been a major prototype for codes with multi-level access \cite{wu2017generalized,zhang2018generalized,blaum2018extended}. GII codes have the advantage of correcting a large set of error patterns, but the distribution of the data symbols is highly restricted, and all the local codewords are equally protected. EII codes are extensions of GII codes with double-level access, where specific arrangements of data symbols have been investigated, mitigating the aforementioned restriction. However, no similar study has been proposed for GII codes with hierarchical locality. Therefore, I-I codes are more suitable for applications where heterogeneity and flexibility are less important. Sum-rank codes are another family of codes that is proposed for dynamic distributed storage offering double-level access\cite{martnez2018universal}. These codes are maximally recoverable, flexible, and allow unequal protection for local data. However, sum-rank codes require a finite field size that grows exponentially with the maximum local block length, which is a major obstacle to being implemented in real world applications.

Hierarchical codes that simultaneously support heterogeneity, scalability, and flexibility are first proposed and investigated in \cite{Yang2019HC} for erasure-resilient cloud storage. However, erasure correction alone is not sufficient to allow these codes being directly used in the computational storage, where both failures (modeled as erasures) and errors can happen. In this paper, we develop an efficient hybrid error correction algorithm of hierarchical codes proposed in \cite{Yang2019HC} to correct both erasures and errors. The paper is organized as follows. In \Cref{section: notation and preliminaries}, we briefly introduce the constructions in \cite{Yang2019HC} and their basic component codes, the EC codes. In \Cref{section: ECneqGRS}, we prove that the class of EC codes is different from existing codes in the literature, and it is strictly larger than the class of generalized Reed Solomon (GRS) codes and generalized Cauchy (GC) codes. In \Cref{section: error correction}, we present an efficient decoding algorithm that corrects a mixture of errors and erasures in EC codes. Based on this algorithm, we develop local and global decoding algorithms for the hierarchical coding scheme introduced in \Cref{section: notation and preliminaries}. Finally, we summarize our results in \Cref{section: conclusion}.

\section{Notation and Preliminaries}
\label{section: notation and preliminaries}{}

Throughout the rest of this paper, $\MB{N}$ refers to $\{1,2,\dots,N\}$, and $\MB{a:b}$ refers to $\{a,a+1,\dots,b\}$. Denote the all zero vector of length $s$ by $\mathbf{0}_s$. Similarly, the all zero matrix of size $s\times t$ is denoted by $\mathbf{0}_{s\times t}$. The alphabet field, denoted by $\textup{GF}(q)$, is a Galois field of size $q$, where $q$ is a power of a prime. For a vector $\mathbf{v}$ of length $n$, $v_i$, $1\leq i\leq n$, represents the $i$-th component of $\mathbf{v}$, and $\mathbf{v}\MB{a:b}=(v_{a},\dots,v_b)$. For a matrix $\mathbf{M}$ of size $a\times b$, $\mathbf{M}\MB{i_1:i_2,j_1:j_2}$ represents the sub-matrix $\mathbf{M}'$ of $\mathbf{M}$ such that $(\mathbf{M}')_{i-i_1+1,j-j_1+1}=(\mathbf{M})_{i,j}$, $i\in\MB{i_1:i_2}$, $j\in\MB{j_1:j_2}$. All indices start from $1$. The operator $\circ$ refers to the Hadamard product, i.e., the element-wise product. The function $\textup{rk}(\cdot)$ returns the rank of a matrix.

\subsection{Extended Cauchy Codes}
\label{subsection: CauchyMatrices}
In this subsection, we introduce the extended Cauchy (EC) codes, which are the major component codes of the hierarchical codes proposed in \cite{Yang2019HC,yang2020topology,yang2020hierarchical}. We start with the definition of the essential ingredients of these codes, the so-called Cauchy matrices, and their extension, generalized Cauchy codes.

\begin{defi} (Cauchy matrix) \label{CauchyMatrix} Let $k,v\in\Db{N}$. Suppose $\textup{GF}(q)$ is a finite field of size $q$. Suppose $a_1,\dots,a_k,b_1,\dots,b_v$ are pairwise distinct elements in $\textup{GF}(q)$. The following matrix is known as a \textbf{Cauchy matrix},
\begin{equation}\label{eqn: CM}\left[
\begin{array}{cccc}
\frac{1}{a_1-b_1} & \frac{1}{a_1-b_2} & \dots & \frac{1}{a_1-b_v}\\
\frac{1}{a_2-b_1} & \frac{1}{a_2-b_2} & \dots & \frac{1}{a_2-b_v}\\
\vdots & \vdots &\ddots & \vdots \\
\frac{1}{a_k-b_1} & \frac{1}{a_k-b_2} & \dots & \frac{1}{a_k-b_v}\\
\end{array}\right].
\end{equation}
We denote this matrix by $\mathbf{Y}(a_1,\dots,a_k;b_1,\dots,b_v)$.
\end{defi}


\begin{defi} \label{defi: gen Cauchy matrix}(generalized Cauchy matrix) \label{defi: genCauchyMatrix} Let $k,v\in\Db{N}$. Suppose $\textup{GF}(q)$ be a finite field of size $q$. Suppose $a_1,\dots,a_k,b_1,\dots,b_v$ are pairwise distinct elements in $\textup{GF}(q)$, and $c_1,\dots,c_k,d_1,\dots,d_v$ are nonzero elements in $\textup{GF}(q)$. The following matrix is known as a \textbf{generalized Cauchy matrix},

\begin{equation}\label{eqn: GCM}\left[
\begin{array}{cccc}
\frac{c_1d_1}{a_1-b_1} & \frac{c_1d_2}{a_1-b_2} & \dots & \frac{c_1d_v}{a_1-b_v}\\
\frac{c_2d_1}{a_2-b_1} & \frac{c_2d_2}{a_2-b_2} & \dots & \frac{c_2d_v}{a_2-b_v}\\
\vdots & \vdots &\ddots & \vdots \\
\frac{c_kd_1}{a_k-b_1} & \frac{c_kd_2}{a_k-b_2} & \dots & \frac{c_kd_v}{a_k-b_v}\\
\end{array}\right].
\end{equation}
\end{defi}

It has been proved in \cite{Yang2019HC} that the Cauchy matrices are essential ingredients in a class of maximum distance separable (MDS) codes, as shown in \Cref{lemma: Good matrix}.

\begin{lemma}(taken from \cite{Yang2019HC}) \label{lemma: Good matrix} Let $k,v,r\in\Db{N}$ such that $v-k<r\leq v$, $\mathbf{A}\in \textup{GF}(q)^{k\times v}$. If $\mathbf{A}$ is a generalized Cauchy matrix, then the following matrix $\mathbf{H}$ is a parity-check matrix of a $(k+r,k+r-v,v+1)_q$-code.
\begin{equation}\label{eqn: EC}
\mathbf{H}=\left[
\begin{array}{c}
\mathbf{A}\\
-\mathbf{I}_r\ \mathbf{0}_{r\times(v-r)}\\
\end{array}
\right]^{\Tx{T}}.
\end{equation}
\end{lemma} 

Note that the condition of $\mathbf{A}$ being a Cauchy matrix in \Cref{lemma: Good matrix} can be relaxed to the more general class of generalized Cauchy matrix described in \Cref{defi: genCauchyMatrix}. In the remaining text, we refer to the code specified by the parity check matrix $\mathbf{H}$ presented in (\ref{eqn: EC}) with $\mathbf{A}$ being a generalized Cauchy matrix as an \textbf{extended Cauchy (EC) code} and denote it by $\Se{C}(\mathbf{A},k,v,r)$. EC codes are the key ingredients of hierarchical coding schemes presented in \cite{Yang2019HC} and \cite{yang2020hierarchical}. 

\begin{rem} Although in the definition of EC codes, we allow the component matrix to be any generalized Cauchy matrix, there exists a map from EC codes onto codes in \Cref{lemma: Good matrix} with $\mathbf{A}$ in (\ref{eqn: EC}) being a Cauchy matrix. In particular, suppose $(x_1,x_2,\dots,x_k,x_{k+1},\dots,x_{k+r})$ is a codeword of an EC code (denote by $\Se{C}_1$) with $\mathbf{A}$ specified as in (\ref{eqn: GCM}). Then, $(c_1x_1,c_2x_2,\dots,c_kx_k,d_1^{-1}x_{k+1},\dots,d_r^{-1}x_{k+r})$ is a codeword of an EC code (denote by $\Se{C}_2$) with $\mathbf{A}$ being a simple Cauchy matrix specified as in (\ref{eqn: CM}). 

For any $\mathbf{c}\in (\textup{GF}(q)\setminus\{0\})^k$ and $\mathbf{d}\in (\textup{GF}(q)\setminus\{0\})^v$, define a map $f: \textup{GF}(q)^n\to \textup{GF}(q)^n$ as follows: 
\begin{equation}
f: (x_1,x_2,\dots,x_k,x_{k+1},\dots,x_{k+r})\mapsto (c_1x_1,c_2x_2,\dots,c_kx_k,d_1^{-1}x_{k+1},\dots,d_r^{-1}x_{k+r}).
\end{equation}
Therefore, $f$ is a bijection from the code $\Se{C}_1$ to the code $\Se{C}_2$. With this property, it is sufficient to consider EC codes with component matrices being a Cauchy matrix in the remaining text.
\end{rem}

\subsection{Hierarchical Coding}
\label{section: codes for multi-level access}

In this subsection, we briefly introduce a construction of the hierarchical codes that support heterogeneity, scalability, and flexibility. For simplicity, we only present the double-level case; the more general cases with higher access levels and those support more flexible error patterns are specified in \cite{Yang2019HC} and \cite{yang2020hierarchical}.

\begin{cons}\label{cons:HC}\cite{Yang2019HC} Let $p\in\Db{N}$, $k_1,k_2,\dots,k_p\in \Db{N}$, $n_1,n_2,\dots,n_p\in\Db{N}$, $\delta_1,\delta_2,\dots,\delta_p\in \Db{N}$, with $r_i=n_i-k_i>\delta_i>0$ for all $i\in\MB{p}$. 
Let $\delta=\sum\nolimits_{i\in \MB{p}}\delta_i$, and suppose $\textup{GF}(q)$ is a Galois field such that $q\geq \max\nolimits_{i\in\MB{p}}\LB{n_i}+\delta$. 
Let $\Se{C}(\mathbf{G})$ represent the code with the generator matrix $\mathbf{G}$ specified as follows:
\begin{equation}\label{eqn: GenMatDL}\mathbf{G}=\left[
\begin{array}{c|c|c|c|c|c|c}
\mathbf{I}_{k_1} & \mathbf{A}_{1,1} & \mathbf{0} & \mathbf{A}_{1,2} & \dots & \mathbf{0} & \mathbf{A}_{1,p}\\
\hline
\mathbf{0} & \mathbf{A}_{2,1} & \mathbf{I}_{k_2} & \mathbf{A}_{2,2} & \dots & \mathbf{0}& \mathbf{A}_{2,p}\\
\hline
\vdots & \vdots & \vdots & \vdots & \ddots & \vdots & \vdots \\
\hline
\mathbf{0} & \mathbf{A}_{p,1} & \mathbf{0} & \mathbf{A}_{p,2} & \dots & \mathbf{I}_{k_p} & \mathbf{A}_{p,p}\\
\end{array}\right],
\end{equation}
where $\mathbf{A}_{i,i}\in \textup{GF}(q)^{k_i\times r_i}$, $\mathbf{B}_{i,j}\in \textup{GF}(q)^{k_i\times \delta_j}$, $\mathbf{U}_i\in \textup{GF}(q)^{\delta_i\times r_i}$ such that $\mathbf{A}_{i,j}=\mathbf{B}_{i,j}\mathbf{U}_j$, for $i,j\in\MB{p}$, $i\neq j$, and there exists Cauchy matrices $\mathbf{T}_i$, $\mathbf{Z}_i$, $i\in\MB{p}$, such that the following equation follows:
\begin{equation}\label{eqn: CRS}
\mathbf{T}_i=\left[
\begin{array}{c|c}
\mathbf{A}_{i,i} & \begin{array}{c|c|c}
\mathbf{B}_{i,1} & \dots & \mathbf{B}_{i,p}
\end{array}
\\
\hline
\mathbf{U}_i & \mathbf{Z}_{i}
\end{array}\right].
\end{equation}
\end{cons}

Codes presented in \Cref{cons:HC} jointly encode local messages $\mathbf{m}_1,\mathbf{m}_2,\dots,\mathbf{m}_p$ to local codewords $\mathbf{c}_1,\mathbf{c}_2,\dots,\mathbf{c}_p$. According to \cite{Yang2019HC}, local codeword $\mathbf{c}_i$ has minimum distance $(r_i-\delta_i+1)$, $i\in\MB{p}$. Moreover, while all codewords in $\{\mathbf{c}_j\}_{j\in\MB{p}\setminus \{i\}}$ are locally correctable, then any hybrid error consisting of $s_i$ errors and $t_i$ erasures in $\mathbf{c}_i$ such that $2s_i+t_i\leq r_i+\delta-\delta_i$ is correctable, $i\in\MB{p}$. In particular, for each $i\in\MB{p}$, suppose $\mathbf{c}_i=(\mathbf{m}_i,\mathbf{s}_i)$; define \textbf{local parity check matrix} $\mathbf{H}_1^{\Tx{L}}$ and \textbf{global parity check matrix} $\mathbf{H}_1^{\Tx{G}}$ as follows:
\begin{equation}\label{eqn: lgPCM}
\mathbf{H}_i^{\Tx{L}}=\left[
\begin{array}{c}
\mathbf{A}_{i,i}\\
\mathbf{U}_i\\
-\mathbf{I}_{r_i}\\
\end{array}
\right]^{\Tx{T}},\ 
\mathbf{H}_i^{\Tx{G}}=\left[
\begin{array}{c|c}
\mathbf{A}_{i,i} & \begin{array}{c|c|c}
\mathbf{B}_{i,1} & \dots & \mathbf{B}_{i,p}
\end{array}
\\
\hline
-\mathbf{I}_{r_i} & \mathbf{0}_{\delta_i\times (\delta-\delta_i)}
\end{array}\right]^{\Tx{T}}.
\end{equation}
Then, $\mathbf{H}_1^{\Tx{L}}$ is the parity check matrix of $\mathbf{c}'_i=(\mathbf{m}_i,\mathbf{0}_{\delta_i},\mathbf{s}_i)$; $\mathbf{H}_1^{\Tx{G}}$ is the parity check matrix of $\mathbf{c}_i$ if other local codewords are corrected.

\begin{table}
\centering
\caption{Polynomial and normal forms of $\textup{GF}(2^4)$}
\begin{tabular}{|c|c||c|c||c|c||c|c|}
\hline
$0$ & $0000$ & $\beta^4$ & $1100$ & $\beta^{8}$ & $1010$ & $\beta^{12}$ & $1111$\\
\hline
$\beta$ & $0100$ & $\beta^{5}$ & $0110$ & $\beta^{9}$ & $0101$ & $\beta^{13}$ & $1011$\\
\hline
$\beta^2$ & $0010$ & $\beta^{6}$ & $0011$ & $\beta^{10}$ & $1110$ & $\beta^{14}$ & $1001$\\
\hline
$\beta^3$ & $0001$ & $\beta^{7}$ & $1101$ & $\beta^{11}$ & $0111$ & $\beta^{15}=1$ & $1000$\\
\hline
\end{tabular}
\label{table: GF}
\end{table}


\begin{exam} \label{exam: CodeDL} Let $q=2^4$, $p=2$, $r=r_1=r_2=3$, $\delta'=\delta_1=\delta_2=1$, $k=k_1=k_2=3$, $n=n_1=n_2=k+r=6$, $\delta=\delta_1+\delta_2=2$. Then, $d_1=r-\delta'+1=3-1+1=3$, $d_2=r-\delta'+\delta+1=3-1+2+1=5$. Choose a primitive polynomial over $\textup{GF}(2)$: $g(X)=X^4+X+1$. Let $\beta$ be a root of $g(X)$, then $\beta$ is a primitive element of $\textup{GF}(2^4)$. The binary representation of all the symbols in $\textup{GF}(2^4)$ is specified in \Cref{table: GF}. 

Let $\mathbf{T}_{1}=\mathbf{T}_{2}=Y(\beta,\beta^2,\beta^3,\beta^4;\beta^8,\beta^9,\beta^{10},\beta^{11})$. Then, the generator matrix $\mathbf{G}$ is specified as follows,
\begin{equation*}
\mathbf{G}=\left[
\begin{array}{ccc|ccc|ccc|ccc}
1 & 0 & 0 & \beta^{5} &\beta^{12} & \beta^{7} & 0 & 0 & 0 & \beta^{4} &\beta^{10} & \beta^{7} \\
0 & 1 & 0 & 1 &\beta^{4} & \beta^{11} & 0 & 0 & 0 & \beta &\beta^{7} & \beta^{4} \\
0 & 0 & 1 & \beta^{2} &\beta^{14} & \beta^{3} & 0 & 0 & 0 & \beta^{5} &\beta^{11} & \beta^{8} \\
\hline
0 & 0 & 0 & \beta^{4} &\beta^{10} & \beta^{7} & 1 & 0 & 0 & \beta^{5} &\beta^{12} & \beta^{7} \\
0 & 0 & 0 & \beta &\beta^{7} & \beta^{4} & 0 & 1 & 0 & 1 &\beta^{4} & \beta^{11} \\
0 & 0 & 0 & \beta^{5} &\beta^{11} & \beta^{8} & 0 & 0 & 1 & \beta^{2} &\beta^{14} & \beta^{3} \\
\end{array}\right].
\end{equation*}
Moreover, $\mathbf{H}_1^{\Tx{L}}$ and $\mathbf{H}_1^{\Tx{G}}$ are specified as follows,
\begin{equation*}
\mathbf{H}_1^{\Tx{G}}=\left[\hspace{-0.1cm}\begin{array}{cccc}
\beta^5 & \beta^{12} & \beta^7 & \beta^9\\
1 & \beta^{4} & \beta^{11} & \beta^6 \\
\beta^2 & \beta^{14} & \beta^3 &\beta^{10}\\
1 & 0 & 0 & 0 \\
0 & 1 & 0 & 0 \\
0 & 0 & 1 & 0
\end{array}\hspace{-0.1cm}\right]^{\Tx{T}},\mathbf{H}_1^{\Tx{L}}=\left[\hspace{-0.1cm}\begin{array}{cccc}
\beta^5 & \beta^{12} & \beta^7 \\
1 & \beta^{4} & \beta^{11}\\
\beta^2 & \beta^{14} & \beta^3\\
\beta^{10} & \beta & \beta^{13}\\
1 & 0 & 0 \\
0 & 1 & 0 \\
0 & 0 & 1 \\
\end{array}\hspace{-0.1cm}\right]^{\Tx{T}}.
\end{equation*}
According to \cite{Yang2019HC}, $\mathbf{G}$ is a generator matrix of a double-level accessible code that has local minimum distance $3$; moreover, when one of the local codewords is locally correctable, the other one tolerates error patterns as if it is from a code with minimum distance $5$. 
\end{exam}

\section{Comparison between EC and GRS/GC Codes}\label{section: ECneqGRS}
EC codes subsume a wide range of Cauchy-like codes. For the trivial case where $r=0$, the parity check matrix is simply a generalized Cauchy matrix and the EC codes degenerate into so-called Cauchy-Reed-Solomon (CRS) codes. CRS codes have received research attention because of their potential in being used in erasure correction for distributed storage systems with low encoding and decoding complexity \cite{bloemer1995xor,plank2006optimizing}.

The wider-known case comes when $r=v$, in which the EC code $\Se{C}(\mathbf{A},k,v,r)$ has a systematic generator matrix of the form $\MB{\mathbf{I}_k\vert\mathbf{A}}$: these codes are referred to as a so-called generalized Cauchy (GC) code \cite{roth1985ongenerator,roth1989onMDScodes}. GC codes have received a wide research attention because of their connections with various advanced codes such as Rabin-like codes and Gabidulin codes \cite{delsarte1978bilinear,gabidulin1985theory}. GC codes have been applied to construct Rabin-like codes with efficient encoding and decoding algorithms in \cite{schindelhauer2013maximum,hou2017new}. Neri \cite{neri2020systematic} has explicitly pointed out that Gabidulin codes can be regarded as the $q$-analogue of GC codes.

It is known that there exists a bijection between the class of GC code and the class of the so-called generalized Reed-Solomon (GRS) codes \cite{dur1991decoding}. GRS codes can be efficiently decoded by variations of the classic Berlekamp-Massey algorithm \cite{dur1991decoding}, which means that all codes $\Se{C}(\mathbf{A},k,v,r)$ with $v=r$ can also be decoded efficiently. 

Surprisingly, we show in \Cref{exam: ECtr1} that EC codes $\Se{C}(\mathbf{A},k,v,r)$ with $v-r=1$ are also GC codes. It is then natural to consider whether the class of EC codes is equivalent to the class of GC codes such that they can be decoded by existing algorithms. Unfortunately, we prove in \Cref{theo: ECneqGRS} that the class of EC codes is strictly larger than that of the GC codes, which means that decoding algorithms for GC codes do not apply for EC codes and new algorithms are needed. Our proof utilizes \Cref{lemma: GRS} that is proposed in \cite{roth1989onMDScodes} to describe the sufficient and necessary conditions for a code to be a GC code. 

\begin{lemma} \label{lemma: GRS} (taken from \cite{roth1989onMDScodes}) Let $\mathbf{X}\in \Db{F}_q^{k\times (n-k)}$. Denote the code specified by the generator matrix $\mathbf{X}$ by $\Se{C}_{\mathbf{X}}$. Then, the code $\Se{C}_{\mathbf{X}}$ is a GRS/GC code if and only if 
\begin{enumerate}
\item every entry $x_{i,j}$ is non-zero,
\item every $2\times 2$ minor of $\mathbf{X}^{c}$ is non-zero, and
\item $\textup{rk}(\mathbf{X}^{c})=2$. 
\end{enumerate} 
Note that $\mathbf{X}^{c}$ refers to the matrix with entries $x^c_{i,j}=x^{-1}_{i,j}$, for $i\in \MB{k}$, $j\in \MB{n-k}$.
\end{lemma}

\begin{exam} \label{exam: ECtr1} In this example, we prove that any EC code $\Se{C}(\mathbf{A},k,v,r)$ with $v-r=1$ is a GRS code. Denote the parity check matrix of  $\Se{C}(\mathbf{A},k,v,r)$ by $\mathbf{H}$ as specified in \Cref{lemma: Good matrix}. Let $\mathbf{Z}=\mathbf{A}\MB{1:k-1,1:r}$, $\mathbf{b}=\mathbf{A}\MB{1:k-1,v}$, $\mathbf{d}^{\textup{T}}=\mathbf{A}\MB{k,1:r}$, and $c=\mathbf{A}\MB{k,v}$. Then the following matrix $\mathbf{L}$ is 
\begin{equation}
\mathbf{L}=\left[\begin{array}{c|c}
\mathbf{0}_{r\times 1} & \mathbf{I}_r\\
\hline
c^{-1} & -c^{-1} \mathbf{d}^{\textup{T}}
\end{array}\right].
\end{equation}
Let $\mathbf{M}=\mathbf{H}\mathbf{L}$. Then, 
\begin{equation}
\mathbf{M}=\left[\begin{array}{c|c}
\mathbf{Z} & \mathbf{b}\\
\hline
\mathbf{d}^{\textup{T}} & c\\
\hline
\mathbf{I}_r & \mathbf{0}_{r\times 1} 
\end{array}\right]\left[\begin{array}{c|c}
\mathbf{0}_{r\times 1} & \mathbf{I}_r\\
\hline
c^{-1} & -c^{-1} \mathbf{d}^{\textup{T}}
\end{array}\right]=\left[\begin{array}{c|c}
\mathbf{b}c^{-1} & \mathbf{Z}-\mathbf{b}c^{-1}\mathbf{d}^{\textup{T}}\\
\hline
1 & \mathbf{0}_{1\times r}\\
\hline
\mathbf{0}_{r\times 1} & \mathbf{I}_r
\end{array}\right].
\end{equation}
Given that $\mathbf{L}$ is nonsingular, $\mathbf{M}$ is also a parity check matrix of the code. Therefore, the code has a systematic generator matrix $\mathbf{G}=\MB{\mathbf{I}_{k-1}\vert\mathbf{X}}$, where $\mathbf{X}=\MB{\mathbf{b}c^{-1}\vert\mathbf{Z}-\mathbf{b}c^{-1}\mathbf{d}^{\textup{T}}}$. Therefore, we only need to prove that this $\mathbf{X}$ satisfies the three conditions specified in \Cref{lemma: GRS}. Note that the three conditions are all invariant under fundamental row and column operations on $\mathbf{X}$. Therefore, it is equivalent to prove the conditions for the matrix $\mathbf{Y}=\MB{\mathbf{Z}-\mathbf{b}c^{-1}\mathbf{d}^{\textup{T}}\vert\mathbf{b}}\in\textup{GF}(q)^{(k-1)\times v}$, for simplicity of indexing.

The elements in $\mathbf{Y}^{c}$ are as follows:
\begin{equation}
y_{i,j}^{-1}=\begin{cases}
\frac{1}{\frac{1}{a_i-b_j}-\frac{a_k-b_v}{(a_i-b_v)(a_k-b_j)}}=\frac{(a_i-b_v)(a_i-b_j)(a_k-b_j)}{(a_i-a_k)(b_v-b_j)}, & 1\leq i<k, 1\leq j<v,\\
a_i-b_v, & 1\leq i<k,j=v.
\end{cases}
\end{equation}
In $\mathbf{Y}^{c}$, we multiply row $i$ by $(a_i-b_v)^{-1}(a_i-a_k)$, $1\leq i<k$, and multiply column $j$, $1\leq j<v$, by $(b_v-b_j)(a_k-b_j)^{-1}$. Then, the resulting matrix $\mathbf{W}$ has $w_{i,j}=a_i-b_j$, for all $1\leq i<k$, $1\leq j<v$, and $w_{i,v}=a_i-a_k$, for all $1\leq i<k$, $j=v$. Therefore, we obtain a Cauchy matrix $\mathbf{W}$ by applying fundamental row and column operations on $\mathbf{Y}^{c}$, which satisfies all the three conditions in the statement of \Cref{lemma: GRS}. According to \Cref{lemma: GRS}, this code is a GRS code.

\end{exam}

We already proved that all EC codes $\Se{C}(\mathbf{A},k,v,r)$ with $v-r\in\{0,1\}$ are GRS codes. However, it does not mean that $\textup{GRS}=\textup{EC}$. In \Cref{theo: ECneqGRS}, we prove that for $(k,v,r)$ with certain constraints, the code $\Se{C}(\mathbf{A},k,v,r)$ can never be a GRS code. 

\begin{theo} \label{theo: ECneqGRS} ($\textup{GC}=\textup{GRS}\subsetneq \textup{EC}$) Any EC code $\Se{C}(\mathbf{A},k,v,r)$ with $q>2v-r$, $v<2r$, $k>2(v-r)+1$, and $v-r\geq 2$, is not a GRS code.
\end{theo}

\begin{proof} Let $\delta=v-r$. Let $\mathbf{Z}=\mathbf{A}\MB{1:k-\delta,1:r}$, $\mathbf{B}=\mathbf{A}\MB{1:k-\delta,r+1:v}$, $\mathbf{D}^{\textup{T}}=\mathbf{A}\allowbreak\MB{k-\delta+1:k,1:r}$, and $\mathbf{C}=\mathbf{A}\MB{k-\delta+1:k,r+1:v}$. Specify matrix $\mathbf{L}$ as follows:
\begin{equation}
\mathbf{L}=\left[\begin{array}{c|c}
\mathbf{0}_{r\times \delta} & \mathbf{I}_r\\
\hline
\mathbf{C}^{-1} & -\mathbf{C}^{-1} \mathbf{D}^{\textup{T}}
\end{array}\right].
\end{equation}
Let $\mathbf{M}=\mathbf{H}\mathbf{L}$. Then, 
\begin{equation}
\mathbf{M}=\left[\begin{array}{c|c}
\mathbf{Z} & \mathbf{B}\\
\hline
\mathbf{D}^{\textup{T}} & \mathbf{C}\\
\hline
\mathbf{I}_r & \mathbf{0}_{r\times \delta} 
\end{array}\right]\left[\begin{array}{c|c}
\mathbf{0}_{r\times \delta} & \mathbf{I}_r\\
\hline
\mathbf{C}^{-1} & -\mathbf{C}^{-1} \mathbf{D}^{\textup{T}}
\end{array}\right]=\left[\begin{array}{c|c}
\mathbf{B}\mathbf{C}^{-1} & \mathbf{Z}-\mathbf{B}\mathbf{C}^{-1}\mathbf{D}^{\textup{T}}\\
\hline
\mathbf{I}_{\delta} & \mathbf{0}_{\delta\times r}\\
\hline
\mathbf{0}_{r\times \delta} & \mathbf{I}_r
\end{array}\right].
\end{equation}
Given that $\mathbf{L}$ is nonsingular, $\mathbf{M}$ is also a parity check matrix of the code. Therefore, the code has a systematic generator matrix of the form $\MB{\mathbf{I}_{k-\delta}\vert\mathbf{X}}$, where $\mathbf{X}=\MB{\mathbf{B}\mathbf{C}^{-1}\vert\mathbf{Z}-\mathbf{B}\mathbf{C}^{-1}\mathbf{D}^{\textup{T}}}$. Similar to the discussion in \Cref{exam: ECtr1}, we only need to prove the three conditions specified in \Cref{lemma: GRS} for matrix $\mathbf{Y}=\MB{\mathbf{Z}-\mathbf{B}\mathbf{C}^{-1}\mathbf{D}^{\textup{T}}\vert\mathbf{B}}\in\textup{GF}(q)^{(k-\delta)\times v}$.

We know that 
\begin{equation}
y_{i,j}^{-1}=\begin{cases}
\frac{1}{\frac{1}{a_i-b_j}-\mathbf{b}_i^{\textup{T}}\mathbf{C}^{-1}\mathbf{d}_j}, &1\leq i\leq k-\delta, 1\leq j\leq r,\\
a_i-b_j, & 1\leq i\leq k-\delta,r<j\leq v,
\end{cases}
\end{equation}
where $\mathbf{b}_i^{\textup{T}}$ denotes the $i$'th row of $\mathbf{B}$, and $\mathbf{d}_j$ denotes the $j$'th column of $\mathbf{D}^{\textup{T}}$.

If Condition 3) in \Cref{lemma: GRS} is satisfied, then, for any pairwise different $i_1,i_2,i_3\in \MB{k-\delta}$ and $j\in\MB{r}$, the following matrix is nonsingular, 
\begin{equation}
\left[\begin{array}{ccc}
\frac{1}{\frac{1}{a_{i_1}-b_j}-\mathbf{b}_{i_1}^{\textup{T}}\mathbf{C}^{-1}\mathbf{d}_j} & a_{i_1}-b_{r+1} & a_{i_1}-b_{r+2}\\
\frac{1}{\frac{1}{a_{i_2}-b_j}-\mathbf{b}_{i_2}^{\textup{T}}\mathbf{C}^{-1}\mathbf{d}_j} & a_{i_2}-b_{r+1} & a_{i_2}-b_{r+2}\\
\frac{1}{\frac{1}{a_{i_3}-b_j}-\mathbf{b}_{i_3}^{\textup{T}}\mathbf{C}^{-1}\mathbf{d}_j} & a_{i_3}-b_{r+1} & a_{i_3}-b_{r+2}\\
\end{array}\right].
\end{equation}
This condition is equivalent to 
\begin{equation}
\begin{split}
&\frac{a_{i_2}-a_{i_3}}{\frac{1}{a_{i_1}-b_j}-\mathbf{b}_{i_1}^{\textup{T}}\mathbf{C}^{-1}\mathbf{d}_j}+\frac{a_{i_3}-a_{i_1}}{\frac{1}{a_{i_2}-b_j}-\mathbf{b}_{i_2}^{\textup{T}}\mathbf{C}^{-1}\mathbf{d}_j}+\frac{a_{i_1}-a_{i_2}}{\frac{1}{a_{i_3}-b_j}-\mathbf{b}_{i_3}^{\textup{T}}\mathbf{C}^{-1}\mathbf{d}_j}=0\\
\Longleftrightarrow &(a_{i_2}-a_{i_3})\SB{\frac{1}{\frac{1}{a_{i_1}-b_j}-\mathbf{b}_{i_1}^{\textup{T}}\mathbf{C}^{-1}\mathbf{d}_j}-\frac{1}{\frac{1}{a_{i_2}-b_j}-\mathbf{b}_{i_2}^{\textup{T}}\mathbf{C}^{-1}\mathbf{d}_j}}\\
=&(a_{i_2}-a_{i_1})\SB{\frac{1}{\frac{1}{a_{i_3}-b_j}-\mathbf{b}_{i_3}^{\textup{T}}\mathbf{C}^{-1}\mathbf{d}_j}-\frac{1}{\frac{1}{a_{i_2}-b_j}-\mathbf{b}_{i_2}^{\textup{T}}\mathbf{C}^{-1}\mathbf{d}_j}}\\
\end{split}
\end{equation}
\begin{equation}
\begin{split}
\Longleftrightarrow  &(a_{i_2}-a_{i_3})\frac{\frac{1}{a_{i_2}-b_j}-\frac{1}{a_{i_1}-b_j}-\SB{\mathbf{b}_{i_2}^{\textup{T}}-\mathbf{b}_{i_1}^{\textup{T}}}\mathbf{C}^{-1}\mathbf{d}_j}{\SB{\frac{1}{a_{i_1}-b_j}-\mathbf{b}_{i_1}^{\textup{T}}\mathbf{C}^{-1}\mathbf{d}_j}\SB{\frac{1}{a_{i_2}-b_j}-\mathbf{b}_{i_2}^{\textup{T}}\mathbf{C}^{-1}\mathbf{d}_j}}\\
=&(a_{i_2}-a_{i_1})\frac{\frac{1}{a_{i_2}-b_j}-\frac{1}{a_{i_3}-b_j}-\SB{\mathbf{b}_{i_2}^{\textup{T}}-\mathbf{b}_{i_3}^{\textup{T}}}\mathbf{C}^{-1}\mathbf{d}_j}{\SB{\frac{1}{a_{i_3}-b_j}-\mathbf{b}_{i_3}^{\textup{T}}\mathbf{C}^{-1}\mathbf{d}_j}\SB{\frac{1}{a_{i_2}-b_j}-\mathbf{b}_{i_2}^{\textup{T}}\mathbf{C}^{-1}\mathbf{d}_j}}\\
\Longleftrightarrow & \frac{(a_{i_2}-a_{i_3})\textup{det}\left(\begin{array}{c|c}
\frac{1}{a_{i_2}-b_j}-\frac{1}{a_{i_1}-b_j} & \mathbf{b}_{i_2}^{\textup{T}}-\mathbf{b}_{i_1}^{\textup{T}}\\
\hline
\mathbf{d}_j & \mathbf{C}
\end{array}\right)}{\textup{det}\left(\begin{array}{c|c}
\frac{1}{a_{i_1}-b_j} & \mathbf{b}_{i_1}^{\textup{T}}\\
\hline
\mathbf{d}_j & \mathbf{C}
\end{array}\right)}\\
=&\frac{(a_{i_2}-a_{i_1})\textup{det}\left(\begin{array}{c|c}
\frac{1}{a_{i_2}-b_j}-\frac{1}{a_{i_3}-b_j} & \mathbf{b}_{i_2}^{\textup{T}}-\mathbf{b}_{i_3}^{\textup{T}}\\
\hline
\mathbf{d}_j & \mathbf{C}
\end{array}\right)}{\textup{det}\left(\begin{array}{c|c}
\frac{1}{a_{i_3}-b_j} & \mathbf{b}_{i_3}^{\textup{T}}\\
\hline
\mathbf{d}_j & \mathbf{C}
\end{array}\right)}\\
\Longleftrightarrow &\frac{\textup{det}\left(\begin{array}{c|c}
\frac{1}{(a_{i_2}-b_j)(a_{i_1}-b_j)} & \mathbf{b}_{i_2}^{\textup{T}}\circ\mathbf{b}_{i_1}^{\textup{T}}\\
\hline
\mathbf{d}_j & \mathbf{C}
\end{array}\right)}{\textup{det}\left(\begin{array}{c|c}
\frac{1}{a_{i_1}-b_j} & \mathbf{b}_{i_1}^{\textup{T}}\\
\hline
\mathbf{d}_j & \mathbf{C}
\end{array}\right)}=\frac{\textup{det}\left(\begin{array}{c|c}
\frac{1}{(a_{i_2}-b_j)(a_{i_3}-b_j)} & \mathbf{b}_{i_2}^{\textup{T}}\circ\mathbf{b}_{i_3}^{\textup{T}}\\
\hline
\mathbf{d}_j & \mathbf{C}
\end{array}\right)}{\textup{det}\left(\begin{array}{c|c}
\frac{1}{a_{i_3}-b_j} & \mathbf{b}_{i_3}^{\textup{T}}\\
\hline
\mathbf{d}_j & \mathbf{C}
\end{array}\right)},\\
\end{split}
\label{eqn: lemma12}
\end{equation}
where $\circ$ denotes the Hadamard product.

Let $\mathbf{b}_X=\MB{\frac{1}{X-b_{r+1}},\frac{1}{X-b_{r+2}},\dots,\frac{1}{X-b_{v}}}$. Define $f_{i,j}(X)\in\textup{GF}(q)(X)$, $i\in\MB{k-\delta}$, $j\in \MB{r}$, as follows,
\begin{equation}
\begin{split}
f_{i,j}(X)\triangleq &\frac{\textup{det}\left(\begin{array}{c|c}
\frac{1}{(a_{i}-b_j)(X-b_j)} & \mathbf{b}_{i}^{\textup{T}}\circ\mathbf{b}_{X}^{\textup{T}}\\
\hline
\mathbf{d}_j & \mathbf{C}
\end{array}\right)}{\textup{det}\left(\begin{array}{c|c}
\frac{1}{X-b_j} & \mathbf{b}_{X}^{\textup{T}}\\
\hline
\mathbf{d}_j & \mathbf{C}
\end{array}\right)}.\\
\end{split}
\label{eqn: lemma11}
\end{equation}
Expanding the determinant of each matrix along its first row, and multiplying both the nominator and the denominator by $(X-b_j)\prod\nolimits_{r<j'\leq v}(X-b_{j'})$, we obtain $f_{i,j}(X)=u_1(X)/u_2(X)$, where $u_1,u_2\in \textup{GF}(q)\MB{X}$ such that $\textup{deg}(u_1),\textup{deg}(u_2)\leq \delta$. Then, (\ref{eqn: lemma12}) implies that there exists $c_1\in \textup{GF}(q)$ such that $f_{i,j}(a_{i'})=c_1$ for all $i'\in\MB{k-\delta}\setminus\{i\}$, $i\in\MB{k-\delta}$, $j\in \MB{r}$, namely, $u_1(X)-c_1u_2(X)=0$ has at least $(k-\delta-1)$ different solutions. Given that $\textup{deg}(u_1-c_1u_2)\leq \delta$, when $k-\delta-1>\delta$, i.e., $k>2\delta+1$, it can be concluded that $u_1(X)-c_1u_2(X= 0$ in $\textup{GF}(q)(X)$. Therefore, $f_{i,j}(X)= c_1$ in $\textup{GF}(q)(X)$, for any $i\in\MB{k-\delta}$, $j\in \MB{r}$.

Find $c\in \textup{GF}(q)$ such that $c$ is different from all $a_{i}$, $k-\delta< i\leq k$, and all $b_{j}$, $1\leq j\leq  v$. When $q>2v-r=v+\delta$, such $c$ always exists. Define $g_{j}(X)\in\textup{GF}(q)(X)$, $j\in\MB{r}$, as follows,
\begin{equation}
\begin{split}
g_{j}(X)\triangleq &\frac{\textup{det}\left(\begin{array}{c|c}
\frac{1}{(X-b_j)(c-b_j)} & \mathbf{b}_{c}^{\textup{T}}\circ\mathbf{b}_{X}^{\textup{T}}\\
\hline
\mathbf{d}_j & \mathbf{C}
\end{array}\right)}{\textup{det}\left(\begin{array}{c|c}
\frac{1}{c-b_j} & \mathbf{b}_{X}^{\textup{T}}\\
\hline
\mathbf{d}_j & \mathbf{C}
\end{array}\right)}.\\
\end{split}
\label{eqn: gj}
\end{equation}

We already concluded that $f_{i,j}(X)= f_{i,j}(c)=g_j(a_i)$ in $\textup{GF}(q)(X)$, for all $i,i'\in \MB{k-\delta}$, $i'\neq i$, $j\in\MB{r}$. Then, (\ref{eqn: lemma12}) is equivalent to the following equation,
\begin{equation}
\frac{f_{i_1,j}(a_{i_2})}{\textup{det}\left(\begin{array}{c|c}
\frac{1}{a_{i_1}-b_j} & \mathbf{b}_{i_1}^{\textup{T}}\\
\hline
\mathbf{d}_j & \mathbf{C}
\end{array}\right)}=\frac{f_{i_3,j}(a_{i_2})}{\textup{det}\left(\begin{array}{c|c}
\frac{1}{a_{i_3}-b_j} & \mathbf{b}_{i_3}^{\textup{T}}\\
\hline
\mathbf{d}_j & \mathbf{C}
\end{array}\right)}\Longleftrightarrow 
\frac{g_{j}(a_{i_1})}{\textup{det}\left(\begin{array}{c|c}
\frac{1}{a_{i_1}-b_j} & \mathbf{b}_{i_1}^{\textup{T}}\\
\hline
\mathbf{d}_j & \mathbf{C}
\end{array}\right)}=\frac{g_{j}(a_{i_3})}{\textup{det}\left(\begin{array}{c|c}
\frac{1}{a_{i_3}-b_j} & \mathbf{b}_{i_3}^{\textup{T}}\\
\hline
\mathbf{d}_j & \mathbf{C}
\end{array}\right)}.
\label{eqn: lemma13}
\end{equation}
Note that (\ref{eqn: lemma13}) holds for any distinct $i_1,i_3$ such that $i_1,i_3\in\MB{k-\delta}$, $i_1\neq i_3$, $j\in\MB{r}$. Define $h_{j}(X)\in\textup{GF}(q)(X)$, $j\in\MB{r}$, as follows: 
\begin{equation}
h_{j}(X)\triangleq\frac{g_{j}(X)}{\textup{det}\left(\begin{array}{c|c}
\frac{1}{X-b_j} & \mathbf{b}_{X}^{\textup{T}}\\
\hline
\mathbf{d}_j & \mathbf{C}
\end{array}\right)}.
\label{eqn: lemma14}
\end{equation}
Expanding the determinant of each matrix in (\ref{eqn: lemma14}) along their first row, and multiplying both the nominator and the denominator of (\ref{eqn: lemma14}) by $(X-b_j)\prod\nolimits_{r<j'\leq v}(X-b_{j'})$, we obtain that $h_{j}(X)=u_3(X)/u_4(X)$, where $u_3,u_4\in \textup{GF}(q)\MB{X}$, $\textup{deg}(u_3),\textup{deg}(u_4)\leq \delta$. Then, (\ref{eqn: lemma14}) implies that there exists $c_2\in \textup{GF}(q)$ such that $h_{j}(a_{i})=c_2$ for all $i\in\MB{k-\delta}$. Namely, $u_3(X)-c_2u_4(X)=0$ has at least $(k-\delta)$ different solutions. Given that $\textup{deg}(u_3-c_2u_4)\leq\delta$, when $k-\delta>\delta$, i.e., $k>2\delta$, it can be concluded that $u_3(X)-c_2u_4(X)= 0$ in $\textup{GF}(q)(X)$. Therefore, $h_{j}(X)= c_2$ in $\textup{GF}(q)(X)$, for any $j\in\MB{r}$.

Then, substitute $h_{j}(X)= c_2$ into (\ref{eqn: lemma11}), (\ref{eqn: gj}), and (\ref{eqn: lemma14}). We conclude that the following equation holds for all $j\in\MB{r}$ in $\textup{GF}(q)(X,Y)$:
\begin{equation}
\begin{split}
\frac{\textup{det}\left(\begin{array}{c|c}
\frac{1}{(X-b_j)(Y-b_j)} & \mathbf{b}_{X}^{\textup{T}}\circ\mathbf{b}_{Y}^{\textup{T}}\\
\hline
\mathbf{d}_j & \mathbf{C}
\end{array}\right)}{\textup{det}\left(\begin{array}{c|c}
\frac{1}{X-b_j} & \mathbf{b}_{X}^{\textup{T}}\\
\hline
\mathbf{d}_j & \mathbf{C}
\end{array}\right)\textup{det}\left(\begin{array}{c|c}
\frac{1}{Y-b_j} & \mathbf{b}_{Y}^{\textup{T}}\\
\hline
\mathbf{d}_j & \mathbf{C}
\end{array}\right)}\textcolor{red}{=} c_2.\\
\end{split}
\label{eqn: lemma15}
\end{equation}

Suppose $\mathbf{C}=\MB{\mathbf{d}_{r+1},\mathbf{d}_{r+2},\dots,\mathbf{d}_{v}}$. In the remainder of this proof, denote the determinant of the matrix obtained from replacing the column $\mathbf{d}_{j_2}$ with $\mathbf{d}_{j_1}$ in $\mathbf{C}$ by $M(j_1,j_2)$, where $1\leq j_1\leq r$ and $r<j_2\leq v$. Denote the determinant of $\mathbf{C}$ by $M$. Then, the condition in (\ref{eqn: lemma15}) is equivalent to that the following equation holds for all $1\leq j_1\leq r$ and $r<j_2\leq v$ in $\textup{GF}(q)(X,Y)$:
\begin{equation}
\frac{\frac{M}{(X-b_{j_1})(Y-b_{j_1})}-\sum\limits_{j_2=r+1}^{v} \frac{M(j_1,j_2)}{(X-b_{j_2})(Y-b_{j_2})} }{\SB{\frac{M}{(X-b_{j_1})}-\sum\limits_{j_2=r+1}^{v} \frac{M(j_1,j_2)}{(X-b_{j_2})}}\SB{\frac{M}{(Y-b_{j_1})}-\sum\limits_{j_2=r+1}^{v} \frac{M(j_1,j_2)}{(Y-b_{j_2})}}}= c_2.
\end{equation}
By multiplying the denominator and the nominator simultaneously with $(X-b_{j_1})(Y-b_{j_1})\allowbreak\prod\limits_{j_2=r+1}^{t=v}\left[(X-b_{j_2})(Y-b_{j_2})\right]$, and assigning $(b_{j_1},b_{j_1}),(b_{r+1},b_{r+1}),\dots,(b_v,b_v)$ to $(X,Y)$, respectively, we obtain $M=c_2^{-1}$ and $M(j_1,j_2)=-c_2^{-1}$, for all $r<j_2\leq v$ and $1\leq j_1\leq r$. Namely, $M(j_1,j_2)=-M$, for all $1\leq j_1\leq r$, $r<j_2\leq v$.

Let $I=\MB{k-\delta+1:k}$, $J=\MB{r+1:v}$. We know from the expression of determinants of Cauchy matrices that
\begin{equation}
\begin{split}
&M=\frac{\prod\limits_{i,i'\in I,i<i'}(a_{i}-a_{i'})\prod\limits_{j,j'\in J,j<j'}(a_{j'}-a_{j})}{\prod\limits_{i\in I,j\in J}(a_{i}-b_{j})},\\
&M(j_1,j_2)=M \prod\limits_{i\in I}\SB{\frac{a_i-b_{j_2}}{a_{i}-b_{j_1}}}\prod\limits_{j\in J\setminus \{j_2\}}\SB{\frac{b_{j}-b_{j_1}}{b_{j}-b_{j_2}}}.
\end{split}
\end{equation}
Therefore, $M(j_1,j_2)=-M$ implies that
\begin{equation}
\begin{split}
&-1= \prod\limits_{i\in I}\SB{\frac{a_i-b_{j_2}}{a_{i}-b_{j_1}}}\prod\limits_{j\in J\setminus \{j_2\}}\SB{\frac{b_{j}-b_{j_1}}{b_{j}-b_{j_2}}}\\
\Longleftrightarrow &\prod\limits_{i\in I}(a_i-b_{j_1})+\frac{\prod\limits_{i\in I}(a_i-b_{j_2})}{\prod\limits_{j\in J\setminus\{j_2\}}(b_j-b_{j_2})}\prod\limits_{j\in J\setminus\{j_2\}}(b_j-b_{j_1})=0.
\end{split}
\label{eqn: lemma16}
\end{equation}
For $j\in J$, define $f_{j}(X)$ as follows:
\begin{equation}
\begin{split}
l_j(X)\triangleq\prod\limits_{i\in I}(a_i-X)+\frac{\prod\limits_{i\in I}(a_i-b_{j})}{\prod\limits_{j'\in J\setminus\{j\}}(b_{j'}-b_{j})}\prod\limits_{j'\in J\setminus\{j\}}(b_{j'}-X).
\end{split}
\label{eqn: lemma17}
\end{equation}
Given that $\Nm{I}=\Nm{J}=\delta$, $\textup{deg}(l_j)=\delta$. The fact that (\ref{eqn: lemma16}) holds for all $j_1$, $1\leq j_1\leq r$, implies that $l_j=0$ has at least $r$ roots, thus $r\leq \delta=v-r$, i.e., $v\geq 2r$. A contradiction. Therefore, any EC code $\Se{C}(\mathbf{A},k,v,r)$ in $\textup{GF}({q})$ with $q>2v-r$, $v<2r$, $k>2(t-r)+1$, $t-r\geq 2$ is not a GRS code.
\end{proof}

We successfully proved that EC codes with special constraints on $(k,v,r)$ cannot be a GRS code. Given that the class of EC codes $\Se{C}(\mathbf{A},k,v,r)$ with $v-r\in\{0,1\}$ is equivalent to the class of GRS codes, we know that $\textup{GRS}\subsetneq \textup{EC}$. Moreover, the constraints imposed on $(k,v,r)$ in \Cref{theo: ECneqGRS} are quite loose just for simplicity, and we do believe that (\ref{eqn: lemma12}) and (\ref{eqn: lemma16}) are hard to satisfy even without the constraints on parameters $(k,v,r)$. This conjecture is left for future investigation.

\section{Hybrid Error Correction of Systematic CRS Codes}
\label{section: error correction}

In \Cref{section: codes for multi-level access}, we briefly recalled the construction of hierarchical codes based on systematic CRS codes proposed in \cite{Yang2019HC}. While discussing erasure correction of these codes, we defined and utilized two parity check matrices $\{\mathbf{H}^{\textup{L}}_i\}_{i=1}^p$, $\{\mathbf{H}^{\textup{G}}_i\}_{i=1}^p$ that belong to EC codes, for local and global decoding, respectively. As proved in \Cref{theo: ECneqGRS}, EC codes cannot be categorized as GRS/GC codes known to be able to decoded by variations of the well known Berlekamp-Massey algorithm. In this section, we propose an efficient error correction algorithm for EC codes, which is necessary and sufficient for both erasure and error correction of the proposed hierarchical codes in computational storage.

Let $n,k,v,r\in\Db{N}$, $0<r\leq v$, $n=k+r$. Suppose $\mathbf{A}\in \textup{GF}(q)^{{k}\times{v}}$ is a Cauchy matrix. Recall that EC code $\Se{C}(\mathbf{A},k,v,r)$ is specified by the parity check matrix $\mathbf{H}$ defined as follows:
\begin{equation}
\mathbf{H}=\left[\begin{array}{ccc}
\mathbf{A}\\
\mathbf{I}_r\ \mathbf{0}_{r\times(v-r)}
\end{array}
\right]^{\textup{T}}.
\label{eqn: H}
\end{equation} 
According to \Cref{lemma: Good matrix}, $\mathbf{H}$ is a parity check matrix of an $(n,n-v,v+1)$-code denoted by $\Se{C}(\mathbf{H})$, which means that $\Se{C}(\mathbf{H})$ is able to correct any hybrid error that is a combination of $s$ errors and $t$ erasures such that $2s+t\leq v$. 

\begin{rem}\label{rem: indexing}For simplicity, let $A=\{a_i\vert i\in \MB{k}\}$, $B=\{b_j\vert j\in\MB{v}\}$. Suppose $\Se{P}(\cdot)$ denotes the power set of a given set. Define $\Se{I}_A:\Se{P}(A)\to \Se{P}(\MB{k})$ by $\Se{I}_A(X)\triangleq\{i\vert a_i\in X\}$. Then, $\Se{I}_A$ is an bijection and $\Se{I}_A^{-1}(I)=\{a_i\vert i\in I\}$. Similarly, define $\Se{I}_B:\Se{P}(B)\to \Se{P}(\MB{v})$ by $\Se{I}_B(X)\triangleq\{j\vert b_j\in X\}$. Then, $\Se{I}_B$ is also an bijection and $\Se{I}_B^{-1}(J)=\{b_j\vert j\in J\}$. For any $J\subset \MB{k+1:n}$, let $J-k=\{j-k\vert j\in J\}$.
\end{rem}

\subsection{Decoding Algorithm for Codes Based on Cauchy Matrices}\label{subsection: decCauchy}
In this subsection, we focus on an efficient decoding algorithm of $\Se{C}(\mathbf{H})$ specified in (\ref{eqn: H}). We define the error location polynomials in \Cref{defi: errortype} and propose a method to efficiently obtain the error location polynomial in \Cref{theo: decDE}. The main algorithm, \Cref{algo: Decode1}, is derived based on \Cref{theo: decDE}, and \Cref{theo: decDE} is proved based on \Cref{lemma: dec} and \Cref{lemma:decconverse}.

\begin{defi}\label{defi: errortype} Let $\textup{GF}(q)$ be a finite field of size $q$. Let $n,k,r\in \mathbb{N}$, where $n=k+r$. Any error on a codeword of length $n$ can be represented by a vector $\mathbf{e}\in(\textup{GF}(q)\cup \{*\})^n$. Let $I(\mathbf{e})=\LB{i:(\mathbf{e})_i\neq 0,i\in\MB{k}}$ and $J(\mathbf{e})=\LB{i:(\mathbf{e})_i\neq 0, i\in\MB{k+1: n}}$. Let $E(\mathbf{e})=\LB{i:(\mathbf{e})_i=*,i\in\MB{n}}$, $D(\mathbf{e})=\LB{i:(\mathbf{e})_i\notin\{0,*\},i\in\MB{n}}$ denote the number of erasures and errors in $\mathbf{e}$, respectively. Then, $I(\mathbf{e})\cup J(\mathbf{e})=E(\mathbf{e})\cup D(\mathbf{e})$.

Define the \textbf{error location polynomial} $e_D(X;\mathbf{e})\in\textup{GF}(q)\MB{X}$, \textbf{erasure location polynomial} $e_E(X;\mathbf{e})\in\textup{GF}(q)\MB{X}$, and the \textbf{hybrid error location polynomial} $e_H(X;\mathbf{e})\in\textup{GF}(q)\MB{X}$ corresponding to an error $\mathbf{e}$ as follows:
\begin{equation}\label{eqn: error location poly}
\begin{split}
e_D(X;\mathbf{e})&\triangleq \Prd{i\in I(\mathbf{e})\setminus E(\mathbf{e})}{}(X-a_i)\Prd{j\in J(\mathbf{e})\setminus E(\mathbf{e})}{}{(X-b_{j-k})},\\
e_E(X;\mathbf{e})&\triangleq \Prd{i\in I(\mathbf{e})\cap E(\mathbf{e})}{}(X-a_i)\Prd{j\in J(\mathbf{e})\cap E(\mathbf{e})}{}{(X-b_{j-k})},\\
e_H(X;\mathbf{e})&\triangleq \Prd{i\in I(\mathbf{e})}{}(X-a_i)\Prd{j\in J(\mathbf{e})}{}{(X-b_{j-k})}.
\end{split}
\end{equation}
Then, $\Tx{deg}\SB{e_E(X;\mathbf{e})}=\Nm{E(\mathbf{e})}$ , $\Tx{deg}\SB{e_D(X;\mathbf{e})}=\Nm{D(\mathbf{e})}$ and $\Tx{deg}\SB{e_H(X;\mathbf{e})}=\Nm{D(\mathbf{e})}+\Nm{E(\mathbf{e})}$. 

For any $s,t\in\Db{N}$, we refer to all errors with $D(\mathbf{e})=s$ and $E(\mathbf{e})=t$ as $\bf{(s,t)}$-errors. For any $v\in\Db{N}$, refer to the union of all $(n,s,t)$-errors with $2s+t\leq v$ as $(n,v)$-errors.
\end{defi}


Note that while any factor $(X-a_i)$ of $e_H(X;\mathbf{e})$, $i\in\MB{k}$, corresponds to an error within the first $k$ elements of the codeword, an error within the last $r$ elements corresponds to $(X-b_{j-k})$ instead, $k<j\leq n$, as indicated in (\ref{eqn: error location poly}). The underlying logic of this definition is that the $j$'th row of $\mathbf{H}^{\textup{T}}$, $k<j\leq n$, which is simply the $(j-k)$'th row of $\MB{\mathbf{I}_r\vert \mathbf{0}_{r\times(v-r)}}$, can be represented in the following way:
\begin{equation}
\left[\begin{array}{ccccccc}
\frac{b_{j-k}-b_{j-k}}{b_{j-k}-b_{1}} & \cdots & \frac{b_{j-k}-b_{j-k}}{b_{j-k}-b_{j-k-1}}  & 1 & \frac{b_{j-k}-b_{j-k}}{b_{j-k}-b_{j-k+1}} & \cdots  & \frac{b_{j-k}-b_{j-k}}{b_{j-k}-b_{v}}
\end{array}\right].
\end{equation}
Define $\mathbf{v}(X)$ as follows:
\begin{equation}\label{eqn: vx}
\mathbf{v}(X)=\left[\begin{array}{cccc}
\frac{1}{X-b_1} & \frac{1}{X-b_{2}} & \cdots & \frac{1}{X-b_{v}}
\end{array}\right].
\end{equation}
Let $\mathbf{v}_{i}(X)=(X-b_i)\mathbf{v}(X)$, for $i\in\MB{v}$. Then, while an error within the first $k$ positions is represented by the vector $\mathbf{v}(a_i)$, for some $i\in\MB{k}$, an error within the last $r$ positions can be represented by the vector $\MB{\frac{X-b_{j-k}}{X-b_1},\frac{X-b_{j-k}}{X-b_{2}},\cdots,\frac{X-b_{j-k}}{X-b_{v}}}$ with $X=b_{j-k}$, for some $k<j\leq n$, which is simply $\mathbf{v}_{j-k}(b_{j-k})$. Therefore, even though the first $k$ columns and the last $r$ columns of the parity check matrix of an EC code are in different formats, our definition of error location polynomials provides a unified representation for all error locations.

For simplicity, denote the $i$'th row of $\mathbf{H}^{\textup{T}}$ specified in (\ref{eqn: H}) by $\mathbf{h}_i$, for $i\in \MB{n}$. For any $W\subset \MB{v}$, denote the subsequence of $\mathbf{h}_i$ containing its $j$'th elements, $j\in W$, by $\mathbf{h}_{i;W}$. Namely, for any $W\subset \MB{v}$, let $u=\Nm{W}$. Suppose $W=\{w_1,w_2,\dots,w_u\}$, $0<w_1<w_2<\cdots<w_u\leq v$, then,
\begin{equation}
\mathbf{h}_{i;W}=\Bigg\{\begin{array}{ccc}
\mathbf{v}_{W}(a_i)=\left[\begin{array}{cccc}
\frac{1}{a_i-b_{w_1}} & \frac{1}{a_i-b_{w_2}} & \cdots & \frac{1}{a_i-b_{w_u}}
\end{array}\right] &,v\in\MB{k},\\
\mathbf{v}_{i-k;W}(b_{i-k})=\left[\begin{array}{cccc}
\frac{b_{i-k}-b_{i-k}}{b_{i-k}-b_{w_1}} & \frac{b_{i-k}-b_{i-k}}{b_{i-k}-b_{w_2}} & \cdots & \frac{b_{i-k}-b_{i-k}}{b_{i-k}-b_{w_u}}
\end{array}\right] &,k<i\leq n.\\
\end{array}
\end{equation}
It is obvious that for any $k<i\leq n$, $h_{i;W}\neq \mathbf{0}_{u}$ if and only if $(i-k)\in W$.

Based on \Cref{defi: errortype}, we are able to define the difference between a codeword and its noisy version, and its syndrome.

\begin{defi}\label{defi: syndrome} Let $\Se{C}(\mathbf{H})$ be a code such that $\mathbf{H}\in \textup{GF}^{v\times n}$. For any codeword $\mathbf{c}\in\Se{C}(H)$, and its noisy version $\mathbf{c}'\in \left(\Tx{GF}(q)\cup\{*\}\right)^n$, the difference of them is an error vector denoted by $\mathbf{e}=\mathbf{c'}-\mathbf{c}$, where for $i\in\MB{n}$,
\begin{equation}
(\mathbf{e})_i=\begin{cases}
(\mathbf{c}')_i-(\mathbf{c})_i&,(\mathbf{c}')_i\neq *,\\
*&,(\mathbf{c'})_i=*.
\end{cases}
\end{equation}
Define the \textbf{zero occupied noisy vector} as the vector obtained by replacing all $*$ in $\mathbf{c}'$ by $0$ and denote it by $z(\mathbf{c}')$. Let $\mathbf{S}=\mathbf{H}z(\mathbf{c}')^{\Tx{T}}$ and refer to $\mathbf{S}$ as the \textbf{syndrome} of $\mathbf{c}'$ associated with $\mathbf{e}$. For any $\mathbf{S}\in\textup{GF}(q)^{v}$ and $\mathbf{e}\in \left(\Tx{GF}(q)\cup\{*\}\right)^n$, $\mathbf{S}$ is called a syndrome associated with $\mathbf{e}$ if there exists some $\mathbf{c}'\in \left(\Tx{GF}(q)\cup\{*\}\right)^n$ such that $\mathbf{S}$ is the syndrome of $\mathbf{c}'$ associated with $\mathbf{e}$.
\end{defi} 

 Observe that $\mathbf{S}$ is a syndrome associated with $\mathbf{e}$ if and only if $\mathbf{S}\in\Tx{Span}\LB{\mathbf{h}_{i}\vert i\in I(\mathbf{e})\cup J(\mathbf{e})}$, i.e., $\mathbf{S}_{W}\in\Tx{Span}\LB{\mathbf{h}_{i;W}\vert i\in I(\mathbf{e})}=\Tx{Span}\{\mathbf{v}_{W\setminus J(\mathbf{e})}(a_i)|i\in I(\mathbf{e})\}$, where $W=\MB{v}\setminus J(\mathbf{e})$. Moreover, for any vector $\mathbf{S}\in\textup{GF}(q)^{v}$, if $\mathbf{S}$ is a syndrome associated with some $(n,s,t)$-error ($(n,v)$-error, resp.) $\mathbf{e}$, then $\mathbf{S}$ is called an $(n,s,t)$-syndrome ($(n,v)$-syndrome, resp.).

In the following \Cref{lemma: dec}, we derive for EC codes a sufficient condition that a syndrome of any error vector must satisfy. While the converse of this lemma does not hold, we prove in \Cref{theo: decDE} a necessary and sufficient condition that the syndrome of any $(n,v)$-error satisfies. \Cref{theo: decDE} provides an efficient method to derive the error location polynomials. Throughout the remaining text, we denote the algebraic closure of $\textup{GF}(q)$ by $K$. According to \cite{langalgebra}, such $K$ always exists.

\begin{lemma}\label{lemma: dec} Let $k,r,v\in\Db{N}$ and $n=k+r$. Let $\mathbf{e}\in \left(\Tx{GF}(q)\cup\{*\}\right)^n$ be an error vector and $\mathbf{S}\in\textup{GF}(q)^v$ be a syndrome associated with $\mathbf{e}$. Let $f\in \textup{GF}(q)\MB{X}\setminus\{0\}$ such that $f$ has no multiplicative roots in $K$ and $e_H(X;\mathbf{e})|f(X)$. Then, $f$ satisfies the following equations in $\textup{GF}(q)$, for all $W\subset\MB{v}$, $\Nm{W}=\Tx{deg}(f)+1$:
\begin{equation}
\Sma{w\in W}{} S_w\frac{f(b_w)}{\Prd{i\in W\setminus\{w\}}{}(b_w-b_i)}=0.
\label{eqn: decmain}
\end{equation}
\end{lemma}

\begin{proof} Let $R\subset K$ be the set consisting of roots of $f(X)=0$ on $K$. The fact that $f$ has no multiplicative roots implies that $\Nm{R}= \Tx{deg}(f)=\Nm{W}-1$. Recall definitions in \Cref{rem: indexing}. Let $I=I(\mathbf{e})$, $J=J(\mathbf{e})-k$, $I_{AR}=\Se{I}_A(R\cap A)$ and $J_{BR}=\Se{I}_B(R\cap B)$. The constraint $e_H(X;\mathbf{e})|f(X)$ indicates that $I\subset I_{AR}$ and $J\subset J_{BR}$.

Let $J_{W}=W\cap J_{BR}$, $B_{W}=\Se{I}_{B}^{-1}(J_{W})$. Let $l=\Nm{W\setminus J_W}$. Suppose $W\setminus J_W=\{w_1,\dots,w_{l}\}$, where $0<w_1<w_2<\cdots<w_l\leq v$. Given that $J_W\subset J_{BR}=\Se{I}_B(R\cap B)$, thus $ B_{W}\subset R\cap B$. Therefore, $J_W\subset W$ implies that $\Nm{R\setminus B_W }=\Nm{R}-\Nm{B_W}= \Nm{W}-1-\Nm{J_W}=\Nm{W\setminus J_W}-1=l-1$. Suppose $R\setminus B_W=\{X_1,\cdots,X_{l-1}\}$. Let $\mathbf{S}_{W\setminus J_W}=\MB{S_{w_1},S_{w_2},\cdots,S_{w_{l}}}$. Define $\mathbf{H}_{W\setminus J_W}$ as follows:
\begin{equation}
 \begin{split}
 &\mathbf{H}_{W\setminus J_W}=\MB{\mathbf{v}_{W\setminus J_W}(X_1)^{\textup{T}},\mathbf{v}_{W\setminus J_W}(X_2)^{\textup{T}},\cdots,\mathbf{v}_{W\setminus J_W}(X_{l-1})^{\textup{T}},\mathbf{S}_{W\setminus J_W}^{\textup{T}}}^{\textup{T}}.
 \end{split}
 \end{equation}
Then,
\begin{equation}
\mathbf{H}_{W\setminus J_W}=\left[\begin{array}{cccc}
\frac{1}{X_1-b_{w_{1}}} & \frac{1}{X_1-b_{w_{2}}} & \cdots & \frac{1}{X_1-b_{w_{l}}} \\
\frac{1}{X_2-b_{w_{1}}} & \frac{1}{X_2-b_{w_{2}}} & \cdots & \frac{1}{X_2-b_{w_{l}}} \\
\vdots & \vdots & \ddots & \vdots\\
\frac{1}{X_{l-1}-b_{w_{1}}} & \frac{1}{X_{l-1}-b_{w_{2}}} & \cdots & \frac{1}{X_{l-1}-b_{w_{l}}} \\
S_{w_{1}} & S_{w_{2}} & \cdots & S_{w_{l}}
\end{array}\right].
\label{eqn: HW}
\end{equation}

The condition that $\mathbf{S}$ is a syndrome associated with $\mathbf{e}$ implies that $\mathbf{S}_{W\setminus J}\in\Tx{Span}\{\mathbf{v}_{W\setminus J}(a_i)|i\in I\}$. Given that $J\subset J_{BR}$, we know that $W\setminus J_W=W\setminus J_{BR}\subset W\setminus J$. Therefore, $\mathbf{S}_{W\setminus J}\in\Tx{Span}\{\mathbf{v}_{W\setminus J}(a_i)|i\in I\}$ implies $\mathbf{S}_{W\setminus J_W}\in\Tx{Span}\{\mathbf{v}_{W\setminus J_W}(a_i)|i\in I\}$. Moreover, provided that $I\subset R\setminus B_W$, this also implies that $\mathbf{S}_{W\setminus J_W}\in\Tx{Span}\{\mathbf{v}_{W\setminus J_W}(x)|x\in R\setminus B_W\}$. Therefore, $\mathbf{H}_{W\setminus J}$ is singular, i.e., $\Tx{det}(\mathbf{H}_{W\setminus J_W})=0$. Expanding $\Tx{det}(\mathbf{H}_{W\setminus J_W})$ along its last row, we obtain:

\begin{equation}
\begin{split}
0=&\Sma{z=1}{l} (-1)^{z-1} S_{w_z}\frac{\Prd{1\leq i<j<l}{}(X_i-X_j)\Prd{1\leq i<j\leq l,i,j\neq z}{}(b_{w_j}-b_{w_i})}{\Prd{1\leq i<l}{}\Prd{1\leq j<l,j\neq z}{}(X_i-b_{w_j})}\\
=&\frac{\Prd{1\leq i<j<l}{}(X_i-X_j)\Prd{1\leq i<j\leq l}{}(b_{w_j}-b_{w_i})}{\Prd{1\leq i<l}{}\Prd{1\leq j\leq l}{}(X_i-b_{w_j})}\Sma{z=1}{l} S_{w_z}\frac{\Prd{1\leq i<l}{}(b_{w_z}-X_i)}{\Prd{1\leq i\leq l,i\neq z}{}(b_{w_z}-b_{w_i})}.\\
\end{split}
\label{eqn:inteq1}
\end{equation}
Given that elements of $R$ are pairwise distinct, $\Prd{1\leq i<j<l}{}(X_i-X_j)\neq 0$. Therefore, 
\begin{equation}
\begin{split}
0=&\Sma{z=1}{l} S_{w_z}\frac{\Prd{1\leq i<l}{}(b_{w_z}-X_i)}{\Prd{1\leq i\leq l,i\neq z}{}(b_{w_z}-b_{w_i})}=\Sma{w\in W\setminus J_W}{} S_{w}\frac{\Prd{x\in R\setminus B_W}{}(b_{w}-x)}{\Prd{w'\in W\setminus J_W,w'\neq w}{}(b_{w}-b_{w'})}\\
=&\Sma{w\in W\setminus J_W}{} S_{w}\frac{\Prd{x\in R}{}(b_{w}-x) }{\Prd{w'\in W\setminus J_W,w'\neq w}{}(b_{w}-b_{w'})\Prd{w'\in J_W}{}(b_{w}-b_{w'})}\\
=&\Sma{w\in W\setminus J_W}{} S_{w}\frac{f(b_w)}{\Prd{w'\in W,w'\neq w}{}(b_{w}-b_{w'})}\\
=&\Sma{w\in W}{} S_{w}\frac{f(b_w)}{\Prd{w'\in W\setminus \{w\}}{}(b_{w}-b_{w'})}.\\
\end{split}
\label{eqn:inteq2}
\end{equation}
The last step holds because $B_W\subset R$, which implies that $f(b_w)=0$ for all $w\in J_W$. The lemma is proved.
\end{proof}

Note that the although the logic flow ``$\Tx{det}(\mathbf{H}_{W\setminus J_W})=0$ for $\mathbf{H}_{W\setminus J_W}$ in (\ref{eqn: HW})''$\Longrightarrow$(\ref{eqn:inteq1})$\Longrightarrow$(\ref{eqn:inteq2}) is invertible provided that $f$ has no multiplicative roots, $\Tx{det}(\mathbf{H}_{W\setminus J_W})=0$ for all $W$ is not sufficient for reaching a conclusion that $\mathbf{S}$ is a syndrome associated with $\mathbf{e}$. That is because the converse of ``$\mathbf{S}_{W\setminus J}\in\Tx{Span}\{\mathbf{v}_{W\setminus J}(a_i)|i\in I\}$ implies $\mathbf{S}_{W\setminus J_W}\in\Tx{Span}\{\mathbf{v}_{W\setminus J_W}(a_i)|i\in I\}$'' is not true. Therefore, the converse of \Cref{lemma: dec} is not true. However, in \Cref{lemma:decconverse}, we prove that as long as $\mathbf{e}$ is an $(n,v)$-error, and constrain $e_E(X;\mathbf{e})|f(X)$, the converse of \Cref{lemma: dec} is true. Moreover, the condition ``all $W$'' can also be loosened by choosing only $\Ceil{\frac{v-t}{2}}$ such $W$'s that satisfy the conditions stated in \Cref{lemma: dec}, where $t$ denotes the number of erasures in $\mathbf{e}$.

\begin{lemma} \label{lemma:decconverse} Let $k,r,v\in\Db{N}$ and $n=k+r$. Let $\mathbf{e}$ be an $(n,v)$-error and $\mathbf{S}$ be a syndrome associated with $\mathbf{e}$. Suppose $t=\Nm{E(\mathbf{e})}$. Let $f\in \textup{GF}(q)\MB{X}\setminus\{0\}$ such that $f$ having no multiplicative roots in $K$, $e_E(X;\mathbf{e})|f(X)$, and $\Tx{deg}(f)=\Floor{\frac{v+t}{2}}$. Let $u=\Tx{deg}(f)$, $W_0\subset\MB{v}$, $\Nm{W_0}=\Tx{deg}(f)=u$. Suppose $\MB{v}\setminus W_0=\{ l_{1},l_{2},\dots,l_{v-u}\}$. Let $W_i=W_0\cup\{l_i\}$, $i\in\MB{v-u}$. Then, $e_H(X;\mathbf{e})|f(X)$ if and only if for all $W\in\LB{W_i}_{i=1}^{v-u}$, the following equation is satisfied:
\begin{equation}
\Sma{w\in W}{} S_w\frac{f(b_w)}{\Prd{i\in W\setminus\{w\}}{}(b_w-b_i)}=0.
\label{eqn: decmain}
\end{equation}
\end{lemma}

\begin{proof} Let $I=I(\mathbf{e})$ and $J=J(\mathbf{e})-k$. Follow the notation in the proof of \Cref{lemma: dec}.

($\Longrightarrow$) Suppose $e_H(X;\mathbf{e})|f(X)$. Then, according to \Cref{lemma: dec}, (\ref{eqn: decmain}) is satisfied.

($\Longleftarrow$) Suppose then (\ref{eqn: decmain}) is satisfied for all $W\in \{W_i\}_{i=1}^{v-u}$. For any $i\in \MB{l-u}$, let $W=W_i$, (\ref{eqn: decmain})$\Longrightarrow$(\ref{eqn:inteq2})$\Longrightarrow$ $\Tx{det}(\mathbf{H}_{W_i\setminus J_{W_i}})=0$ for $\mathbf{H}_{W_i\setminus J_{W_i}}$ in (\ref{eqn:inteq1}) holds, which implies that $\mathbf{S}_{W_i\setminus J_{W_i}}\in\Tx{Span}\LB{\mathbf{v}_{W_i\setminus J_{W_i}}(x)|x\in R\setminus B_{W_i}}$. Therefore, there exists $\{\alpha_{i}(x)\}_{x\in R\setminus B_{W_i}}\subset K$, such that 
\begin{equation}\label{eqn: converse1}
\mathbf{S}_{W_i\setminus J_{W_i}}=\sum\limits_{x\in R\setminus B_{W_i}}\alpha_i(x) \mathbf{v}_{W_i\setminus J_{W_i}}(x).
\end{equation} 


Let $J_{W_0}=W_0\cap J_{BR}$ and $B_{W_0}=\Se{I}_{B}^{-1}(J_{W_0})$. Then, $R\setminus B_{W_i} \subset R\setminus B_{W_0}$, and $W_i\setminus J_{W_i}=W_i\setminus J_{BR}$ gives that $W_0\setminus J_{W_0}\subset W_i\setminus J_{W_i}$, for all $i\in\MB{l-u}$. Let $\alpha_i(b_{l_i})=0$ for $l_i\in J_{BR}$. Therefore, (\ref{eqn: converse1}) implies that
\begin{equation}\label{eqn: converse2}
\mathbf{S}_{W_0\setminus J_{W_0}}=\sum\limits_{x\in R\setminus B_{W_0}}\alpha_i(x) \mathbf{v}_{W_0\setminus J_{W_0}}(x).
\end{equation}
Compute the difference of (\ref{eqn: converse2}) with two different $i,i'\in \MB{l-u}$ it follows that 
\begin{equation}\label{eqn: converse3}
\mathbf{0}_{\Nm{W_0\setminus J_{W_0}}}=\sum\limits_{x\in R\setminus B_{W_0}}(\alpha_i(x)-\alpha_{i'}(x)) \mathbf{v}_{W_0\setminus J_{W_0}}(x).
\end{equation}
Note that $\Nm{R\setminus B_{W_0}}=\Nm{W_0\setminus J_{W_0}}$, and $\mathbf{v}_{W_0\setminus J_{W_0}}(x)$, $x\in R\setminus B_{W_0}$, are linearly independent, (\ref{eqn: converse3}) implies that $\alpha_i(x)=\alpha_{i'}(x)$ for all $x\in R\setminus B_{W_0}$. Therefore, there exists $\alpha(x)$, $x\in R\setminus B_{W_0}$, such that $\alpha_i(x)=\alpha(x)$, for all $i\in\MB{l-u}$. Moreover, $\alpha(b_{l_i})=0$, for all $l_i\in J_{BR}$, i.e., $\alpha(x)=0$, for all $x\in R\cap B$. Then, (\ref{eqn: converse1}) implies that for all $j\in \bigcup\nolimits_{i=1}^{l-u} \SB{W_i\setminus J_{W_i}}$, $S_{j}=\sum\limits_{x\in R\cap A}\alpha(x) \mathbf{v}_{\{j\}}(x)$. Since $\bigcup\nolimits_{i=1}^{l-u} \SB{W_i\setminus J_{W_i}}=\bigcup\nolimits_{i=1}^{l-u}\SB{W_i\setminus J_{BR}}=\SB{\bigcup\nolimits_{i=1}^{l-u} W_i}\setminus J_{BR}=\MB{v}\setminus J_{BR}$. This means that $\mathbf{S}_{\MB{v}\setminus J_{BR}}=\sum\limits_{x\in R\cap A}\alpha(x) \mathbf{v}_{\MB{v}\setminus J_{BR}}(x)\in\Tx{Span}\SB{\mathbf{v}_{\MB{v}\setminus J_{BR}}(x)|x\in R\cap A}$, i.e., $\mathbf{S}\in\Tx{Span}\LB{\mathbf{h}_{i}|i\in I_{AR}\cup J_{BR}}$. 

Therefore, there exists $\{\alpha_i\}$, $i\in \MB{n}$ such that $\alpha_i=0$ for all $i\notin I_{AR}\cup J_{BR}$, and
\begin{equation}\label{eqn: converse5}
\mathbf{S}=\sum\limits_{i\in \MB{n}}\alpha_i \mathbf{h}_i.
\end{equation}
Provided that $\mathbf{S}$ is a syndrome associated with $\mathbf{e}$, there exists $\{\beta_i\}_{x\in \MB{n}}$, such that $\beta_i=0$ for all $i\notin I\cup J$, and
\begin{equation}\label{eqn: converse6}
\mathbf{S}=\sum\limits_{i\in \MB{n}}\beta_i \mathbf{h}_i.
\end{equation}
Subtracting (\ref{eqn: converse6}) from (\ref{eqn: converse5}), we obtain 
\begin{equation}\label{eqn: converse7}
\mathbf{0}_v=\sum\limits_{i\in \MB{n}}(\alpha_i-\beta_i) \mathbf{h}_i=\sum\limits_{i\in I_{AR}\cup J_{BR}\cup I \cup J}(\alpha_i-\beta_i) \mathbf{h}_i.
\end{equation}
Let $E=E(\mathbf{e})$. Then, $t=|E|$, $E\subset (I\cup J)$, and $|(I\cup J)\setminus E|=|D(\mathbf{e})|\leq \Floor{\frac{v-t}{2}}$. Provided that $e_E(X;\mathbf{e})|f(X)$, $E\subset (I_{AR}\cup J_{BR})$, and $|(I_{AR}\cup J_{BR})\setminus E|=|I_{AR}\cup J_{BR}|-|E|\leq |R|-|E|=\Floor{\frac{v-t}{2}}$. Therefore, $|I_{AR}\cup J_{BR}\cup I \cup J|=\Nm{E\cup \SB{(I\cup J)\setminus E}\cup \SB{(I_{AR}\cup J_{BR})\setminus E}}\leq t+\Floor{\frac{v-t}{2}}+\Floor{\frac{v-t}{2}}\leq v$. It follows that $\alpha_i=\beta_i$, for all $i\in I_{AR}\cup J_{BR}\cup I \cup J$. Since $\beta_i\neq 0$, for all $i\in I\cup J$, it implies that $(I\cup J)\subset (I_{AR}\cup J_{BR})$, which is equivalent to $e_H(X;\mathbf{e})|f(X)$.
The lemma is proved.
\end{proof}

For any $(n,v)$-error $\mathbf{e}$ of an EC code $\Se{C}(\mathbf{A},k,v,r)$, \Cref{lemma:decconverse} presents a necessary and sufficient criteria to determine whether a given multiple of $e_E(X;\mathbf{e})$ is also a multiple of $e_H(X;\mathbf{e})$. Based on this condition, \Cref{theo: decDE} provides an efficient method to obtain a multiple of $e_H(X;\mathbf{e})$ provided the erasure locations and the syndrome. Note that in the last few steps of \Cref{lemma:decconverse}, the two most critical arguments are $|I_{AR}\cup J_{BR}\cup I \cup J|=\Nm{E\cup \SB{(I\cup J)\setminus E}\cup \SB{(I_{AR}\cup J_{BR})\setminus E}}\leq t+u+u\leq v$, and $\Nm{\bigcup\nolimits_{i=1}^{l-u} W_i}=v$, where $u=\Floor{\frac{v-t}{2}}$. Following a similar logic, if we know the upper bound on the number of errors to be some $s_0<\Floor{\frac{v-t}{2}}$, then the degree of $f$ and the cardinality of $W_0$ in the condition of \Cref{lemma:decconverse} and \Cref{theo: decDE} can be replaced by $(s_0+t)$ accordingly. With this new condition, it is sufficient to select $(2s_0+t)$ different $W_i$'s, e.g., $\{W_i\}_{i=1}^{2s_0+t}$, and the same conclusion follows. 

\begin{theo}\label{theo: decDE} Let $\textup{GF}(q)$ be a finite field of size $q$. Let $n,k,r,v\in \mathbb{N}$, where $n=k+r$ and $r\leq v$. Let $\Se{C}(\mathbf{A},k,v,r)$ be an EC code. Suppose $\mathbf{c}$ is a codeword and $\mathbf{c}'$ is a noisy version of $\bf{c}$ such that $\mathbf{e}=\mathbf{c'}-\mathbf{c}$ is an $(n,v)$-error. Let $\mathbf{S}$ be the syndrome of $\mathbf{c}'$ associated with $\mathbf{e}$.

Let $t=\Nm{E(\mathbf{e})}$, $s=\Floor{\frac{v-t}{2}}$, and $s'=\Ceil{\frac{v-t}{2}}$. Suppose $W_0\subset \MB{v}$, where $|W_0|=s+t$. Then, $\Nm{\MB{v}\setminus W_0}=v-s-t=s'$. Suppose $\MB{v}\setminus W_0=\{w_i\}_{i=1}^{s'}$, and let $W_i=W_0\cup\{w_i\}$, for $1\leq i\leq s'$. Let $\mathbf{F}\in\textup{GF}(q)^{s'\times s}$, $\mathbf{u}\in\textup{GF}(q)^{s}$. Let $\tilde{\mathbf{F}}=\MB{\mathbf{u}|\mathbf{F}}$ be specified as follows: 
\begin{equation}
(\tilde{\mathbf{F}})_{i,j}=
\Sma{w\in W_i}{} \frac{S_w e_E(b_w;\mathbf{e})b_w^{s+1-j}}{\Prd{i\in W_i\setminus\{w\}}{}(b_w-b_i)}= \frac{S_{w_i}e_E(b_{w_i};\mathbf{e})b_{w_i}^{s+1-j}}{\Prd{i\in W_0}{}(b_{w}-b_i)}+\Sma{w\in W_0}{} \frac{S_we_E(b_w;\mathbf{e})b_w^{s+1-j}}{(b_{w}-b_{w_i})\Prd{i\in W_0\setminus\{w\}}{}(b_w-b_i)}.
\label{eqn: F}
\end{equation}

Let $g(X;\bm{\sigma})\triangleq X^{s}+\sigma_1X^{s-1}+\sigma_2 X^{s-2}+\cdots+\sigma_{s-1}X+\sigma_{s}$, $\bm{\sigma}\in\textup{GF}(q)^{s}$, be an arbitrary monic polynomial that has no multiplicative roots over $K$ and satisfies $\textup{gcd}(g,e_E)=1$. Then, $e_D(X;\mathbf{e})|g(X;\bm{\sigma})$ if and only if $\mathbf{F}\bm{\sigma}=-\mathbf{u}$.
\end{theo}

\begin{proof} Let $f(X)=g(X;\bm{\sigma})e_E(X;\mathbf{e})$. The conditions $\textup{gcd}(g,e_E)=1$ and $g$ has no multiplicative roots over $K$ are equivalent to $f$ has no multiplicative roots over $K$ and $e_E(X;\mathbf{e})|f(X;\bm{\sigma})$. Moreover, $f(w)=g(w;\bm{\sigma})e_E(w;\mathbf{e})$, for $w\in\MB{v}$. Then, equation $\mathbf{F}\bm{\sigma}=-\mathbf{u}$ is equivalent to (\ref{eqn: decmain}) in \Cref{lemma:decconverse}. Therefore, $e_H(X;\mathbf{e})|f(X;\bm{\sigma})$ according to \Cref{lemma:decconverse}, which implies $e_D(X;\mathbf{e})|g(X;\bm{\sigma})$. The theorem is proved.
\end{proof}

\Cref{theo: decDE} leads to \Cref{algo: Decode1}, an efficient decoding algorithm of EC codes. The essential step is to obtain a multiple $g(X;\bm{\sigma})$ of the error locating polynomial by solving the equation $\mathbf{F}\bm{\sigma}=-\mathbf{u}$ with $\mathbf{F},\mathbf{u}$ specified in (\ref{eqn: F}). Preparing the matrix $\tilde{\mathbf{F}}$ and solving the equation $\mathbf{F}\bm{\sigma}=-\mathbf{u}$ both need $O(s^3)$ operations (multiplications and additions of numbers) over $\textup{GF}(q)$. Therefore, a multiple of $e_H(X;\mathbf{e})$ can be obtained with complexity $O((s+t)^3)$. After that, locating the error errors, i.e., finding the elements of $A\setminus E(\mathbf{e})$ and $B\setminus E(\mathbf{e})$ that are the roots of $g(X)$ over $\textup{GF}(q)$, requires $O(s^2+ns)$ operations if we simply do brute force search by assigning each element of $(A\cup B)\setminus E(\mathbf{e})$ to $g(X)$ and check if it results in a zero function value. After that, $O((s+t)^2)$ operations are needed to obtain the correct values on all the error positions. Therefore, this process has an overall complexity $O((s+t)^3+ns)$.

Up to here we are very close to the final decoding algorithm. Cautious readers might have noticed that there is still one potential slack in \Cref{theo: decDE}: $e_Eg$ must have no multiplicative roots over $K$, i.e., $g$ has no multiplicative roots and $\textup{gcd}(e_{\textup{E}},g)=1$. \Cref{lemma: Seperablepoly} implies that if $\bm{\sigma}_1,\bm{\sigma}_2$ are linearly independent, then for any $Y\subset\textup{GF}(q)$ such that $|Y|=2(s+t)+1$, $\{\bm{\sigma}_{\gamma}=\gamma\bm{\sigma}_1+(1-\gamma) \bm{\sigma}_2\}_{\gamma\in Y}$ contains at least one element $\bm{\sigma}_{\gamma}$ such that $g(X;\bm{\sigma}_{\gamma})$ is satisfiable. Moreover, it is known that $g$ has multiplicative root over its splitting field $K$ if and only if $\textup{deg}(\textup{gcd}(g',g))>0$ in $\textup{GF}(q)\MB{X}$ \cite{langalgebra}. For each $g(X;\bm{\sigma})$, to obtain $\textup{gcd}(g',g)$ and $\textup{gcd}(g,e_E)$, one only needs an extra complexity of $O(s^2+st)$ by Euclidean algorithm. Therefore, it takes $O(s(s+t)^2)$ operations in the worst case to find a satisfiable $\bm{\sigma}$ in the worst case. In fact, the additional complexity is only $O(s(s+t))$ on average. We now can safely announce that the overall complexity is $O((s+t)^3+ns)$.

\begin{algorithm}
\caption{Decoding Algorithm for Code $\Se{C}(\mathbf{H})$}\label{algo: Decode1}
\begin{algorithmic}[1]
\Require 
\Statex $\mathbf{c}'$: received noisy codeword;
\Statex $s$: error correction limit;
\Statex $t$: erasure correction limit;
\Statex $\mathbf{H}$: parity check matrix in (\ref{eqn: H});
\Ensure
\Statex $\hat{\mathbf{c}}$: estimation of the transmitted codeword;
\State Find the set $E$ consisting of locations of erasures in $\mathbf{c}'$ and $e_E$;
\State $\mathbf{S}\gets z(\mathbf{c}')\mathbf{H}^{\textup{T}}$;
\State Find $W_0\subset \MB{2s+t}$, $|W_0|=s+t$. Suppose $\MB{v}\setminus W_0=\{w_i\}_{i=1}^{v-s-t}$;
\State Obtain $\tilde{\mathbf{F}}$ according to (\ref{eqn: F}), $\mathbf{F}\gets \tilde{\mathbf{F}}\MB{:,1:s+t}$, $\mathbf{u}\gets \tilde{\mathbf{F}}\MB{:,0}$;
\If{$\mathbf{F}\mathbf{\bm{\sigma}}=-\mathbf{u}$ has an unique solution $\bm{\sigma}$ (up to multiplication with a scalar)}
\State Find $I_0\subset \MB{k}$ and $J_0\subset \MB{r}$ such that elements in $\Se{I}_A^{-1}(I_0)\cup \Se{I}_B^{-1}(J_0)$ are roots of $g(X;\bm{\sigma})$;
\State Find $\mathbf{e}$ such that $\sum\nolimits_{i\in I_0\cup (J_0+k)\cup E} (\mathbf{e})_i \mathbf{h}_{i}=\mathbf{S}$; $\hat{\mathbf{c}}\gets z(\mathbf{c'})+\mathbf{e}$;
\State \textbf{return} $\hat{\mathbf{c}}$;
\ElsIf{$\mathbf{F}\mathbf{\bm{\sigma}}=\mathbf{u}$ has at least two linearly independent solutions $\bm{\sigma}_1,\bm{\sigma}_2$}
\State Pick up a random set $Y\subset \textup{GF}(q)\setminus\{1\}$, $|Y|=2(s+t)$;
\For{$\gamma\in Y$}
$\bm{\sigma}\gets \gamma\bm{\sigma}_1+(1-\gamma)\bm{\sigma}_2$, $g\gets g(X;\bm{\sigma})$;
\State Run Euclidean Algorithm to find $h_1=\textup{gcd}(g,g')$, $h_2=\textup{gcd}(g,e_E)$;
\If{$\textup{deg}(h_1)=0$ and $\textup{deg}(h_2)=0$}
\State Repeat step 6-8; 
\EndIf

\EndFor
\EndIf
\end{algorithmic}
\end{algorithm}

\begin{lemma}\label{lemma: Seperablepoly} Let $n\in\Db{N}$, $\bm{\sigma}_1,\bm{\sigma}_2\in \textup{GF}(q)^{s}$, and $g(X)\in \textup{GF}(q)\MB{X}$. Suppose $Y\subset\textup{GF}(q)$ such that $|Y|=2(s+\textup{deg}(g))+1$. Define for any $\bm{\sigma}\in\textup{GF}(q)$,
\begin{equation}
f(X;\bm{\sigma})=g(X)(X^{s}+\sigma_1X^{s-1}+\sigma_2 X^{s-2}+\cdots+\sigma_{s-1}X+\sigma_{s}).
\end{equation}
 Then, there exists $\gamma\in Y$ such that $f(X;\bm{\sigma})$ has no multiplicative roots over its splitting field $K$, where $\bm{\sigma}_{\gamma}=\gamma\bm{\sigma}_1+(1-\gamma)\bm{\sigma}_2$.
\end{lemma}

\begin{proof}
It is obvious that $|Y|=2\Tx{deg}(f)+1$. We know that the discriminant of $f(X;\bm{\sigma})$ is a function of $\sigma_1,\cdots,\sigma_{s}$ with degree at most $2\Tx{deg}(f)$, thus we can denote it by $\Delta({\bm{\sigma}})$. Define $l(Z)\in \textup{GF}(q)[Z]$ as follows:
\begin{equation}
l(Z)\triangleq \Delta((1-Z)\bm{\sigma}_1+Z\bm{\sigma}_2).
\end{equation}
Then, $\Tx{deg}(l)\leq 2\Tx{deg}(f)$ given that each $\sigma_i$ is a linear function of $Z$, which means that $g(Z)$ has at most $2\Tx{deg}(f)$ roots. Therefore, there exists $\gamma\in Y$ such that $l(\gamma)\neq 0$, which is equivalent to $\Delta(\bm{\sigma}_{\gamma})\neq 0$. Therefore $f(X;\bm{\sigma}_{\gamma})$ has no multiplicative roots over its splitting field $K$.
\end{proof}

We next show an example of our decoding method being applied to an EC code, which is also a GRS/GC code.

\begin{exam} Following the normal forms of elements in $\textup{GF}(2^4)$ specified in \Cref{table: GF}, consider the GRS/GC code with the generator matrix $\left[\mathbf{I}|\mathbf{A}\right]$, where $\mathbf{A}$ is a generalized Cauchy matrix with $\mathbf{a}=(\beta,\beta^2,\beta^3,\beta^4)$, $\mathbf{b}=(\beta^5,\beta^6,\beta^7,\beta^8,\beta^9)$, $\mathbf{c}=(\beta^{10},\beta^{11},\beta^{12},\beta^{13})$, and $\mathbf{d}=(\beta^{2},\beta^{4},\beta^{13},\beta^{14},1)$. Follow \Cref{cons:HC}, the generator matrix $\mathbf{G}$ is obtained as follows:
\begin{equation}
\mathbf{G}=\left[
\begin{array}{ccccc}
\frac{\beta^{10+2}}{\beta+\beta^5} & \frac{\beta^{10+4}}{\beta+\beta^6} & \frac{\beta^{10+13}}{\beta+\beta^7} & \frac{\beta^{10+14}}{\beta+\beta^8} & \frac{\beta^{10}}{\beta+\beta^9} \\
\frac{\beta^{11+2}}{\beta^2+\beta^5} & \frac{\beta^{11+4}}{\beta^2+\beta^6} & \frac{\beta^{11+13}}{\beta^2+\beta^7} & \frac{\beta^{11+14}}{\beta^2+\beta^8} & \frac{\beta^{11}}{\beta^2+\beta^9} \\
\frac{\beta^{12+2}}{\beta^3+\beta^5} & \frac{\beta^{12+4}}{\beta^3+\beta^6} & \frac{\beta^{12+13}}{\beta^3+\beta^7} & \frac{\beta^{12+14}}{\beta^3+\beta^8} & \frac{\beta^{12}}{\beta^3+\beta^9} \\
\frac{\beta^{13+2}}{\beta^4+\beta^5} & \frac{\beta^{13+4}}{\beta^4+\beta^6} & \frac{\beta^{13+13}}{\beta^4+\beta^7} & \frac{\beta^{13+14}}{\beta^4+\beta^8} & \frac{\beta^{13}}{\beta^4+\beta^9} \\
\end{array}\right]=\left[
\begin{array}{ccccc}
\beta^{10} & \beta^{3} & \beta^{9} & \beta^{14} & \beta^{7} \\
\beta^{12} & \beta^{12} & \beta^{12} & \beta^{10} & 1 \\
\beta^{3} & \beta^{14} & \beta^{6} & \beta^{13} & \beta^{11} \\
\beta^{7} & \beta^{5} & \beta^{8} & \beta^{7} & \beta^{14} \\
\end{array}\right].
\end{equation}

A message vector $\mathbf{m}=(1,1,1,1)$ is then encoded into $\mathbf{c}=(1,1,1,1,\beta^7,\beta^3,\beta^6,\beta^3,\beta^{13})$. Suppose there are two errors in $\mathbf{c}$ such that $\mathbf{c}'=(1,0,1,1,\beta^7,\beta^3,\beta^6,\beta^3,\beta^6)$. We then obtain the syndrome $\mathbf{S}=(\beta^{12},\beta^{12},\beta^{12},\beta^{10},0)$. Consider the map $f: (x_1,x_2,x_3,x_4,x_5,x_6,x_7,x_8,x_9)\mapsto(\beta^{10}x_1,\beta^{11}x_2,\beta^{12}x_3,\beta^{13}x_4,\beta^{-2}x_5,\beta^{-4}x_6,\beta^{-13}x_7,\beta^{-14}x_8,x_9)$, and $g: (s_1,s_2,s_3,s_4,s_5)\mapsto(\beta^{-2}s_1,\allowbreak\beta^{-4}s_2,\beta^{-13}s_3,\beta^{-14}s_4,s_5)$. We obtain the equivalent codeword $\tilde{\mathbf{c}}=f(\mathbf{c}')=(\beta^{10},0,\beta^{12},\beta^{13},\beta^5,\allowbreak\beta^{14},\beta^{8},\beta^4,\beta^{6})$, and the syndrome $\tilde{\mathbf{S}}=g(\mathbf{S})=(\beta^{10},\beta^8,\beta^{14},\beta^{11},0)$.

Let $W_0=\{1,5\}$, $W_1=\{1,2,5\}$, $W_2=\{1,3,5\}$. Let $g(X;\bm{\sigma})=X^2+\sigma_1 X+\sigma_2$. Then, 
\begin{equation*}
\begin{split}
0=&\beta^{10}\frac{\beta^{5*2}+\beta^{5}\sigma_1+\sigma_2}{(\beta^{5}-\beta^{6})(\beta^{5}-\beta^{9})}+\beta^{8}\frac{\beta^{6*2}+\beta^{6}\sigma_1+\sigma_2}{(\beta^{6}-\beta^{5})(\beta^{6}-\beta^{9})}+0\frac{\beta^{9*2}+\beta^{9}\sigma_1+\sigma_2}{(\beta^{9}-\beta^{5})(\beta^{9}-\beta^{6})},\\
0=&\beta^{10}\frac{\beta^{5*2}+\beta^{5}\sigma_1+\sigma_2}{(\beta^{5}-\beta^{7})(\beta^{5}-\beta^{9})}+\beta^{14}\frac{\beta^{7*2}+\beta^{7}\sigma_1+\sigma_2}{(\beta^{7}-\beta^{5})(\beta^{7}-\beta^{9})}+0\frac{\beta^{9*2}+\beta^{9}\sigma_1+\sigma_2}{(\beta^{9}-\beta^{5})(\beta^{9}-\beta^{7})},
\end{split}
\end{equation*}
which is equivalent to
\begin{equation}\left[
\begin{array}{cc}
0 & \beta^3\\
\beta^{11} & 1
\end{array}\right]\left[\begin{array}{c}
\sigma_{1}\\
\sigma_{2}
\end{array}\right]=\left[\begin{array}{c}
\beta^{14}\\
\beta^{8}
\end{array}\right].
\end{equation}
The solution is $(\sigma_1,\sigma_2)=(\beta^{11},\beta^{11})$, thus $f(X)=X^2+\beta^{11}X+\beta^{11}=(X-\beta^2)(X-\beta^{9})=(X-a_2)(X-b_5)$. Therefore, $\mathbf{e}=(0,e_1,0,0,0,0,0,0,e_2)$, and $e_1(\beta^{12},\beta^{12},\beta^{12},\beta^{10},1)+e_2(0,0,0,0,1)=(\beta^{12},\beta^{12},\beta^{12},\beta^{10},0)$. We know that $e_1=1$ and $e_2=1$. Therefore, $\hat{\mathbf{c}}=\mathbf{c}'-\mathbf{e}=(1,1,1,1,\beta^7,\beta^3,\beta^6,\beta^3,\beta^{13})$.

\end{exam}

\subsection{Local Decoding Algorithm}\label{subsection: localdecoding}
Based on the decoding algorithm for $s$-error-$t$-erasure codes with the parity check matrix specified by (\ref{eqn: H}), the local decoding process is specified in \Cref{algo: LocalDecodeAlg}. The key idea is briefly stated as follows. Suppose $\mathbf{c}_i=(\mathbf{m}_i,\mathbf{s}_i)$, and $\mathbf{p}_i=\sum\nolimits_{j\neq i}\mathbf{m}_j\mathbf{B}_{j,i}$, for all $i,j\in\MB{p}$ such that $j\neq i$. Provided that $\mathbf{s}_i=\mathbf{m}_i\mathbf{A}_{i,i}+\sum\nolimits_{j\neq i}\mathbf{m}_j\mathbf{B}_{j,i}\mathbf{U}_i=\mathbf{m}_i\mathbf{A}_{i,i}+\mathbf{p}_i\mathbf{U}_i$, $\mathbf{H}^{\Tx{L}}_i$ is the parity check matrix of an $(n_i+\delta_i,k_i,r_i+\delta_i+1)$-EC code such that $\tilde{c}_i=(\mathbf{m}_i,\mathbf{p}_i,\mathbf{s}_i)$ is a codeword. In local decoding, $\mathbf{p}_i$ is not known but the positions of its symbols are known, thus $\mathbf{p}_i$ can be regarded as $\delta_i$ erasures in $\tilde{c}_i$. Since $\Se{C}(\mathbf{H}^{\Tx{L}}_i)$ has minimum distance $r_i+\delta_i+1$, it corrects any extra $ \lfloor\frac{r-\delta}{2}\rfloor$ errors in $\mathbf{c}_i$.

\begin{algorithm}
\caption{Local Decoding Algorithm} \label{algo: LocalDecodeAlg}
\begin{algorithmic}[1]
\Require 
\Statex $\mathbf{c}_i'$: received noisy codeword in the $i$-th block;
\Statex $\mathbf{H}_i^{\textup{L}}$: local parity check matrix in (\ref{eqn: lgPCM});
\Ensure
\Statex $\hat{\mathbf{m}}_i$: estimation of the message in the $i$-th block;
\State $s\gets \lfloor\frac{r-\delta}{2}\rfloor$, $t\gets \delta$, $\mathbf{H}\gets \mathbf{H}_i^{\textup{L}}$;
\State $\mathbf{m}'_i\gets \mathbf{c}'_i\MB{1:k_i}$, $\mathbf{u}'_i\gets\mathbf{c}'_i\MB{k_i+1:n_i}$;
\State $\mathbf{c}'\gets \MB{\mathbf{m}'_i,\mathbf{*}^{\delta},\mathbf{u}'_i}$, where $*$'s refers to erasures;
\State Run \Cref{algo: Decode1} and obtain $\hat{\mathbf{c}}$;
\State $\hat{\mathbf{m}}_i\gets \hat{\mathbf{c}}\MB{1:k_i}$;
\State \textbf{return} $\hat{\mathbf{m}}_i$;
\end{algorithmic}
\end{algorithm}

\begin{exam} (Local Decoding) Use the code with parameters specified in \Cref{exam: CodeDL}. Let $\mathbf{m}=(\mathbf{m}_1,\mathbf{m}_2)$, where $\mathbf{m}_1=(\beta,0,\beta^4)$, $\mathbf{m}_2=(0,1,0)$. Then, $\mathbf{c}_1=(\beta,0,\beta^4,\beta,\beta^{11},\beta^{13})$. Let $\mathbf{c}'_1=(\beta,\beta^2,\beta^4,\beta,\beta^{11},\beta^{13})$, then $d_H(\mathbf{c}_1,\mathbf{c}'_1)=1$, and thus $\bm{c}_1$ can be locally decoded. 

Following \Cref{algo: LocalDecodeAlg}, we obtain $s=\Floor{\frac{r-\delta}{2}}=\Floor{\frac{3-1}{2}}=1$, $t=\delta=1$, $\mathbf{m}'_1=(\beta,\beta^2,\beta^4)$, $\mathbf{u}'_i=(\beta,\beta^{11},\beta^{13})$ and $\mathbf{c}'=(\beta,\beta^2,\beta^4,*,\beta,\beta^{11},\beta^{13})$, and $z(\mathbf{c}')=(\beta,\beta^2,\beta^4,0,\beta,\beta^{11},\beta^{13})$. Then, in \Cref{algo: Decode1}, $E=\{4\}$, $e_E(X)=X-\beta^4$, and $\mathbf{S}=(S_1,S_2,S_3)=z(\mathbf{c}')(\mathbf{H}^{\Tx{L}}_1)^{\textup{T}}=(\beta^5,\beta^{10},\beta^{11})$, where $\mathbf{H}^{\Tx{L}}_1$ is specified in \Cref{exam: CodeDL}. Moreover, we know that $v=3$ and $v-s-t=1$, thus we only need to focus on a single $W$, $W=\MB{v}=\{1,2,3\}$, in \Cref{algo: Decode1}.

Then, $g(X;\bm{\sigma})=X-\sigma_1$. Provided that $b_1=\beta^8$, $b_2=\beta^9$, and $b_3=\beta^{10}$, (\ref{eqn: F}) implies
\begin{equation*}
\begin{split}
0=&\beta^{5}\frac{(\beta^{8}-\beta^{4})(\beta^{8}-\sigma_1)}{(\beta^{8}-\beta^{9})(\beta^{8}-\beta^{10})}+\beta^{10}\frac{(\beta^{9}-\beta^{4})(\beta^{9}-\sigma_1)}{(\beta^{9}-\beta^{8})(\beta^{9}-\beta^{10})}+\beta^{11}\frac{(\beta^{10}-\beta^{4})(\beta^{10}-\sigma_1)}{(\beta^{10}-\beta^{8})(\beta^{10}-\beta^{9})},
\end{split}
\end{equation*}
which is equivalent to the following equation by multiplying each side of the equation by $(\beta^8-\beta^9)(\beta^9-\beta^{10})(\beta^{10}-\beta^8)$:
\begin{equation*}
0=\beta^{8}(\beta^8-\sigma_1)+\beta^{10}(\beta^9-\sigma_1)+\beta^{10}(\beta^{10}-\sigma_1).
\end{equation*}
Therefore, $\sigma_1=\frac{\beta+\beta^{4}+\beta^{5}}{\beta^8+\beta^{10}+\beta^{10}}=\frac{\beta^{10}}{\beta^{8}}=\beta^2$. Then, $g(X;\bm{\sigma})=X-\sigma_1=X-a_2$, thus the error vector $\mathbf{e}=(0,e_2,0,e_4,0,0,0)$, and 
\begin{equation*}
\left[\begin{array}{cc}
e_2 & e_4\\
\end{array}\right]\left[\begin{array}{ccc}
1 & \beta^4 & \beta^{11}\\
\beta^{10} & \beta & \beta^{13}
\end{array}\right]=\left[\begin{array}{ccc}
\beta^5 & \beta^{10} & \beta^{11}\\
\end{array}\right].
\end{equation*}
We then obtain $({e}_2,{e}_4)=(\beta^2,\beta^{6})$. Therefore, $\hat{\mathbf{c}}_1=(0,\beta^2,0,0,0,0)+\mathbf{c}'_1=(\beta,0,\beta^4,\beta,\beta^{11},\beta^{13})$, and $\hat{\mathbf{m}}_1=(\beta,0,\beta^4)$.
\end{exam}

\subsection{Global Decoding Algorithm}\label{subsection: globaldecoding}
In this subsection, we present the global decoding algorithm in \Cref{algo: GlobalDecodeAlg} for the code specified in \Cref{section: codes for multi-level access}. The core idea of \Cref{algo: GlobalDecodeAlg} is to obtain extra syndromes from the locally recoverable blocks. For simplicity, assume other blocks are all decoded. Let $\mathbf{S}_{i,k}=(\mathbf{m}'_i-\mathbf{m}_i)\mathbf{B}_{i,k}$ and $\mathbf{p}_{i,k}=\sum\nolimits_{j\neq k,i}\mathbf{m}_j\mathbf{A}_{j,k}$, for $i,k\in \MB{p}$ such that $k\neq i$. Then, $\mathbf{s}_k=\mathbf{m}_k\mathbf{A}_{k,k}+\mathbf{m}'_i\mathbf{A}_{i,k}-\mathbf{S}_{i,k}\mathbf{U}_i+\mathbf{p}_{k,i}\mathbf{U}_i$. Provided $\mathbf{U}_i$ is a full row rank matrix and $\mathbf{S}_{i,k}\mathbf{U}_i=\mathbf{m}_k\mathbf{A}_{k,k}-\mathbf{s}_k+\mathbf{m}'_i\mathbf{A}_{i,k}+\mathbf{p}_{i,k}$, $\mathbf{S}_{i,k}$ can be obtained. Moreover, the local syndrome $\mathbf{S}_i=(\mathbf{s}_i-\mathbf{m}_i\mathbf{A}_{i,i})-(\mathbf{s}'_i-\mathbf{m}'_i\mathbf{A}_{i,i})=\mathbf{m}'_i\mathbf{A}_{i,i}-\mathbf{s}'_i+\sum\nolimits_{j\neq i}\mathbf{m}_j\mathbf{A}_{j,i}$. Then $(\mathbf{S}_i,\mathbf{S}_1,\cdots,\mathbf{S}_{i-1},\mathbf{S}_{i+1},\cdots,\mathbf{S}_m)$ is the syndrome of $\mathbf{c}'_i$ corresponding to the global parity check matrix $\mathbf{H}^{\Tx{G}}_i$ in (\ref{eqn: lgPCM}).

 We start with an example.
\begin{algorithm}
\caption{Global Decoding Algorithm} \label{algo: GlobalDecodeAlg}
\begin{algorithmic}[1]
\Require 
\Statex $i$: the index of the message needs to be globally decoded;
\Statex $\mathbf{c}'$: received noisy codeword (assume $\mathbf{c}'_j$, $j\neq i$, has been all locally corrected, i.e., $\mathbf{c}'_j=\mathbf{c}_j$);
\Statex $\mathbf{H}_i^{\textup{G}}$: global parity check matrix in (\ref{eqn: lgPCM});
\Ensure
\Statex $\hat{\mathbf{m}}_i$: estimation of the message in the $i$-th block;
\State $s\gets \lfloor\frac{r+(m-1)\delta}{2}\rfloor$, $t\gets 0$, $\mathbf{H}\gets\mathbf{H}_i^{\textup{G}}$;
\State $\mathbf{S}_i\gets \mathbf{m}'_i\mathbf{A}_{i,i}-\mathbf{s}'_i+\sum\nolimits_{j\neq i}\mathbf{m}_j\mathbf{A}_{j,i}$, where $\mathbf{c}_j=(\mathbf{m}_j,\mathbf{s}_j)$;
\For{$k\in \MB{m}\setminus\{i\}$}
\State $\tilde{\mathbf{S}}_k\gets\mathbf{m}_k\mathbf{A}_{k,k}-\mathbf{s}_k+\mathbf{m}'_i\mathbf{A}_{i,k}+\sum\nolimits_{j\neq k,i}\mathbf{m}_j\mathbf{A}_{j,k}$, where $\mathbf{c}_j=(\mathbf{m}_j,\mathbf{s}_j)$;
\State Solve $\mathbf{S}_k\mathbf{U}_k=\tilde{\mathbf{S}}_k$;
\EndFor
\State $\mathbf{S}\gets \MB{\mathbf{S}_i,\mathbf{S}_1,\cdots,\mathbf{S}_{i-1},\mathbf{S}_{i+1},\cdots,\mathbf{S}_m}$, $\mathbf{c}'\gets \mathbf{c}'_i$;
\State Run \Cref{algo: Decode1} and obtain $\hat{\mathbf{c}}$;
\State $\hat{\mathbf{m}}_i\gets \hat{\mathbf{c}}\MB{1:k_i}$;
\State \textbf{return} $\hat{\mathbf{m}}_i$;
\end{algorithmic}
\end{algorithm}

\begin{exam} (Global Decoding) Use the code constructed in \Cref{exam: CodeDL}. Let $\mathbf{m}=(\mathbf{m}_1,\mathbf{m}_2)$, where $\mathbf{m}_1=(\beta,0,\beta^4)$, $\mathbf{m}_2=(0,1,0)$. Therefore, $\mathbf{c}=(\mathbf{c}_1,\mathbf{c}_2)$, where $\mathbf{c}_1=(\beta,0,\beta^4,\beta,\beta^{11},\beta^{13})$, and $\mathbf{c}_2=(0,1,0,\beta^{13},\beta^6,\beta^2)$. Suppose $\mathbf{c}'=(\mathbf{c}'_1,\mathbf{c}_2)$, where $\mathbf{c}'_1=(\beta,1,\beta^4,\beta,\beta^9,\beta^{13})$. Then, $\mathbf{m}'_1=(\beta,1,\beta^4)$, $\mathbf{s}'_1=(\beta,\beta^9,\beta^{13})$, $\mathbf{m}_2=(0,1,0)$, and $\mathbf{s}_2=(\beta^{13},\beta^6,\beta^2)$. We obtain $\mathbf{S}_1=\mathbf{s}'_1-\mathbf{m}'_1\mathbf{A}_{1,1}-\mathbf{m}_2\mathbf{A}_{2,1}=(1,\beta^{10},\beta^{11})$, and $\tilde{\mathbf{S}}_2=\mathbf{s}_2-\mathbf{m}_2\mathbf{A}_{2,2}-\mathbf{m}'_1\mathbf{A}_{1,2}=(\beta,\beta^{7},\beta^{4})$. Given that $(\beta,\beta^{7},\beta^{4})=\beta^{6}(\beta^{10},\beta,\beta^{13})=\beta^{6}\mathbf{U}_2$, we obtain $\mathbf{S}=(1,\beta^{10},\beta^{11},\beta^{6})$. 

Let $W_0=\{2,3\}$, $W_1=\{1,2,3\}$, $W_2=\{2,3,4\}$, and $g(X;\bm{\sigma})=X^2+\sigma_1X+\sigma_2$. Provided that $b_1=\beta^8$, $b_2=\beta^9$, $b_3=\beta^{10}$, and $b_4=\beta^{11}$, (\ref{eqn: F}) implies
\begin{equation*}
\begin{split}
0=&1\frac{\beta^{8*2}+\beta^{8}\sigma_1+\sigma_2}{(\beta^{8}-\beta^{9})(\beta^{8}-\beta^{10})}+\beta^{10}\frac{\beta^{9*2}+\beta^{9}\sigma_1+\sigma_2}{(\beta^{9}-\beta^{8})(\beta^{9}-\beta^{10})}+\beta^{11}\frac{\beta^{10*2}+\beta^{10}\sigma_1+\sigma_2}{(\beta^{10}-\beta^{8})(\beta^{10}-\beta^{9})},\\
0=&\beta^{10}\frac{\beta^{9*2}+\beta^{9}\sigma_1+\sigma_2}{(\beta^{9}-\beta^{10})(\beta^{9}-\beta^{11})}+\beta^{11}\frac{\beta^{10*2}+\beta^{10}\sigma_1+\sigma_2}{(\beta^{10}-\beta^{9})(\beta^{10}-\beta^{11})}+\beta^{6}\frac{\beta^{11*2}+\beta^{11}\sigma_1+\sigma_2}{(\beta^{11}-\beta^{9})(\beta^{11}-\beta^{10})},
\end{split}
\end{equation*}
which is equivalent to
\begin{equation}\left[
\begin{array}{cc}
\beta & \beta^{5}\\
\beta^{6}&\beta^{2}
\end{array}\right]\left[\begin{array}{c}
\sigma_{1}\\
\sigma_{2}
\end{array}\right]=\left[\begin{array}{c}
\beta^{13}\\
\beta^{14}
\end{array}\right].
\end{equation}
The solution is $(\sigma_1,\sigma_2)=(\beta^{11},\beta^{11})$, thus $f(X)=X^2+\beta^{11}X+\beta^{11}=(X-\beta^2)(X-\beta^9)=(X-a_2)(X-b_2)$. Therefore, $\mathbf{e}_1=(0,e_1,0,0,e_2,0)$, and $e_1(1,\beta^4,\beta^{11},\beta^6)+e_2(0,1,0,0)=(1,\beta^{10},\beta^{11},\beta^6)$. We know that $e_1=1$, $e_2=\beta^{10}-\beta^4=\beta^2$. Therefore, $\mathbf{e}_1=(0,1,0,0,\beta^2,0)$, and $\hat{\mathbf{c}}_1=\mathbf{c}'_1-\mathbf{e}=(\beta,0,\beta^4,\beta,\beta^{11},\beta^{13})$.
\end{exam}

\section{Conclusion}
\label{section: conclusion}
Computational storage has garnered substantial research interests for its ability to significantly reduce the latency by moving data-processing down to the data storage, which is critical for intelligent devices in the IoT ecosystem. ECCs with are indispensable to protect the stored data against errors. To meet the aggressive latency requirements of intelligent devices at the edge, hierarchical codes that are heterogeneous, scalable, and flexible are desired. While our prior work in hierarchical codes for erasure-resiliency in cloud storage already meets the aforementioned properties, we developed in this paper an efficient decoding algorithm that corrects a mixture of errors and erasures simultaneously such that these codes are also applicable to computational storage. We first proved that EC codes, the major component codes in the proposed construction, do not belong to relevant existing codes in the family of RS codes or Cauchy codes with known explicit decoding algorithms. We then presented an efficient decoding method for the general class of EC codes. Based on the decoding algorithm, we proposed the local and global decoding algorithms tailored for the proposed hierarchical codes. Future work includes extending the construction such that global access enables error-correction of concurrent multiple local access failures.

\section*{Acknowledgment}

This work was supported in part by UCLA Dissertation Year Fellowship, NSF under the Grants CCF-BSF 1718389, CCF 1717602, and CCF 1908730, and in part by AFOSR under the Grant 8750-20-2-0504.

\bibliography{ref}
\bibliographystyle{IEEEtran}

\end{document}